\documentclass[journal,comsoc]{IEEEtran}

\makeatletter
\newcommand\semihuge{\@setfontsize\semihuge{22.3}{22}}
\makeatother

\usepackage{tabularx}

\usepackage{algpseudocode}
\usepackage{algorithm}
\usepackage{algorithmicx}

\usepackage{lipsum} 
\usepackage{arydshln}
\usepackage[dvips]{color}
\usepackage{comment}
\usepackage{todonotes}
\usepackage{epsf}
\usepackage{epsfig}
\usepackage{times}
\usepackage{epsfig}
\usepackage{graphicx}
\usepackage{bbold}
\usepackage{bbm}
\usepackage{mathtools}
\usepackage{mathrsfs}
\usepackage{amssymb}
\usepackage{pdfpages}
\usepackage{epstopdf}
\newfloat{algorithm}{t}{lop}

\usepackage{amsmath}
\usepackage{dsfont}
\usepackage{lettrine} 

\usepackage{amsmath,epsfig,amssymb,algorithm,algpseudocode,amsthm,cite,url}
\newtagform{blue}{\color{blue}(}{)}
\usepackage{subcaption}
\allowdisplaybreaks
\usepackage{csquotes}

\usepackage{verbatim}
\usepackage[english]{babel}
\usepackage{amsmath,amssymb}
\usepackage{diagbox}

\DeclareMathOperator*{\argmax}{arg\,max}

\usepackage[justification=centering]{caption}
\usepackage{verbatim}

\newtheorem{corollary}{Corollary}
\newtheorem{theorem}{\bf Theorem}

\newtheorem{proposition}{\bf Proposition}
\newtheorem{lemma}{\bf Lemma}

\newtheorem{definition}{\bf Definition}

\begin{document}
	%
	\title{Deep Learning for Signal Authentication and Security in Massive Internet of Things Systems}
	\IEEEoverridecommandlockouts
	\author{\IEEEauthorblockN{Aidin Ferdowsi\IEEEauthorrefmark{1} and Walid Saad\IEEEauthorrefmark{1}\\}
		\IEEEauthorblockA{\IEEEauthorrefmark{1}
			Wireless@VT, Bradley Department of Electrical and Computer Engineering, \\ Virginia Tech, Blacksburg, VA, USA,
			Emails: \{aidin,walids\}@vt.edu.\\}\vspace{-1cm}
		\thanks{This research was supported by the U.S. National Science Foundation under Grants OAC-1541105, CNS-1446621, and CNS-1524634. A preliminary version of this work appears in the proceedings of IEEE ICC 2018 \cite{IoT2018ferdowsi}.}
		
	}
	\maketitle
	
	\maketitle
	

	%
	\IEEEpeerreviewmaketitle
	
\begin{abstract}	
Secure signal authentication is arguably one of the most challenging problems in the Internet of Things (IoT), due to the large-scale nature of the system and its susceptibility to man-in-the-middle and {data injection} attacks. In this paper, a novel { watermarking algorithm }is proposed for dynamic authentication of IoT signals to detect cyber attacks. {The proposed watermarking algorithm, based on a deep learning long short-term memory (LSTM) structure}, enables the IoT devices (IoTDs) to extract a set of stochastic features from their generated signal and dynamically watermark these features into the signal. This method enables the IoT gateway, which collects signals from the IoTDs, to effectively authenticate the reliability of the signals. Moreover, in massive IoT scenarios, since the gateway cannot authenticate all of the IoTDs simultaneously due to computational limitations, a game-theoretic framework is proposed to improve the gateway's decision making process by predicting vulnerable IoTDs. The mixed-strategy Nash equilibrium (MSNE) for this game is derived and the uniqueness of the expected utility at the equilibrium is proven. In the massive IoT system, due to the large set of available actions for the gateway, the MSNE is shown to be analytically challenging to derive, and, thus, a learning algorithm that converges to the MSNE is proposed. Moreover, in order to handle incomplete information scenarios in which the gateway cannot access the state of the unauthenticated IoTDs, a deep reinforcement learning algorithm is proposed to dynamically predict the state of unauthenticated IoTDs and allow the gateway to decide on which IoTDs to authenticate. Simulation results show that, with an attack detection delay of under 1 second, the messages can be transmitted from IoTDs with an almost $ 100\% $ reliability. The results also show that, by optimally predicting the set of vulnerable IoTDs, the proposed deep reinforcement learning algorithm reduces the number of compromised IoTDs by up to $30\%$, compared to an equal probability baseline.
\end{abstract}
\section{Introduction}
The Internet of Things (IoT) will encompass a massive number of devices that must reliably transmit a environmental observations to deliver a plethora of smart city applications \cite{IoT2017Dawy,IoT2017Abuzainab,IoT2015Islam,UAV2016Mozaffari,IoT2014Xu}. However, an effective deployment of IoT services requires near real-time, secure, and low complexity message transmission from the IoT devices (IoTDs) \cite{ali2018fast}. Most IoT architectures consist of four layers: perceptual, network, support, and application\cite{security2012Suo}. The perceptual layer is the most fundamental layer which collects all sorts of information from the physical world using devices such as accelerometers or radio frequency identification (RFID) tags. Due to the simplicity of the devices and components at the perceptual layer and their resource-constrained nature, securing the IoT signals at this layer is notoriously challenging\cite{security2012Suo}. 

Recently, a number of security solutions have been proposed for IoT signal authentication \cite{Physical2015Mukherjee, physical2015Trappe, lightweight2014Lee, security2016Yaman, mahalle2013identity, TURKANOVIC201496, Butun2016}. The work in \cite{Physical2015Mukherjee} investigated physical layer security techniques for securing IoT applications. These methods include optimal sensor censoring, channel-based bit flipping, probabilistic ciphering of quantized IoT signals, and artificial noise signal transmission. In \cite{physical2015Trappe}, the author suggested bridging the security gap in IoTDs by applying information theory and cryptography
at the physical layer of the IoT. An authentication protocol for the IoT is presented in \cite{lightweight2014Lee}, using lightweight encryption method in order to cope with constrained IoTDs. Moreover, in \cite{security2016Yaman}, the authors developed a learning mechanisms for fingerprinting and authenticating IoTDs and their environment. Other useful signal authentication approaches are found in \cite{mahalle2013identity,TURKANOVIC201496,Butun2016}.

To secure IoT-like cyber-physical systems (CPSs), the idea of watermarking has been studied in \cite{watermark2015Mo,Watermarking2017Satchidanandan,hespanhol2017dynamic,watermark2017Hosseini}. In \cite{watermark2015Mo}, a new method was proposed to watermark a predefined signal unto a CPS input signal so as to detect replay attacks. A dynamic watermarking algorithm was proposed in \cite{Watermarking2017Satchidanandan} for integrity attack detection in networked CPSs. In \cite{hespanhol2017dynamic}, the authors introduced a security scheme that ensures detection of data injection attacks using a non-stationary watermarking technique. Finally, the authors in \cite{watermark2017Hosseini} analyzed the optimality of Gaussian watermarked signals in presence of cyber attacks in linear time-variant IoT-like systems. Moreover, the security of massive IoT systems has gained attention in recent years \cite{2017Dorri,2012Sehgal,sun2017,2018Stergiou}. In \cite{2017Dorri}, a blockchain-based approach was proposed to provide a distributed security solution for the IoT. In \cite{2012Sehgal}, the authors demonstrated the feasibility of implementing existing device management protocols on resource-constrained devices. A \emph{cloud}-based algorithm was proposed in \cite{sun2017} to provide privacy service for resource-constrained IoTDs. Finally, the work in \cite{2018Stergiou} investigated how cloud computing can be securely integrated in large-scale IoT systems.

However, the authentication solutions of \cite{Physical2015Mukherjee,physical2015Trappe,lightweight2014Lee,security2016Yaman,mahalle2013identity,TURKANOVIC201496,Butun2016} remain highly complex for deployment at the IoT perceptual layer, and require high computational power. Moreover, these methods do not take into account dynamic data injection attacks in which the attacker collects data for a long time duration and uses it for designing stealthy attacks. Furthermore, the watermarking algorithms introduced in \cite{watermark2015Mo,Watermarking2017Satchidanandan,hespanhol2017dynamic,watermark2017Hosseini}, can be detrimental to the performance of a system such as the IoT since an augmented watermark is applied in parallel to the control signal of the system. This can, in turn, lead to a suboptimal performance for the system. In addition, the input signals to IoTDs include information such as temperature, heart rate, and location, which are not controllable and require changing the IoTDs' environment. Moreover, the works in \cite{2017Dorri,2012Sehgal,sun2017,2018Stergiou} do not consider resource constraints at the IoT gateway for authentication. In a large-scale IoT system, the gateway cannot authenticate all of the transmitted signals from the IoTDs due to the large amount of required computational resources and bandwidth. Therefore, in a practical IoT, the gateway must optimally and intelligently choose which IoTDs to authenticate.

The main contribution of this paper is a comprehensive framework that integrates new ideas from deep learning and game theory to enable computationally-efficient authentication of IoT signals and devices, in massive IoT systems. The key contributions include:
\begin{itemize}
	\item We propose a novel {watermarking} framework that enables the IoT's gateway to authenticate IoTD signals and detect the existence of a cyber attacker who seeks to degrade the performance of the IoT by changing the devices' output signal. The proposed { watermarking} algorithm uses {deep} long short-term memory (LSTM) \cite{chen2017machine} blocks to extract stochastic features such as spectral flatness, skewness, kurtosis, and central moments from IoT signal and watermarks these features inside the original signal. 
	\item This dynamic feature extraction approach allows the gateway to detect { \emph{dynamic data injection attacks} in which the attacker can record and process IoT signals, extract the watermarking key, and inject faulty data}. Moreover, the proposed {LSTM-based watermarking} effectively complements other security solutions, such as encryption, by reducing the complexity and latency of data injection attack detection.
	\item To enable the gateway to authenticate IoTDs in a massive IoT system and under resource constraints, we formulate a noncooperative game between the gateway and the attacker to derive the gateway's optimal action in predicting vulnerable IoTDs while considering its computational resource constraints. We derive the mixed-strategy Nash equilibrium of the game and prove the uniqueness of the gateway's expected utility at this equilibrium.
	\item We show that analytically finding the gateway's optimal strategy is computationally expensive in the massive IoT scenario. Therefore, we propose a learning algorithm based on \emph{fictitious play} that converges to the mixed-strategy Nash equilibrium. Finally, we {extend our analysis to} a practical case in which the gateway does not have complete information about all the IoTDs. To address this challenge we propose a deep reinforcement learning approach based on LSTM blocks, to learn the {security} state of the IoTDs based on their past states. We show that the gateway's expected utility using this proposed deep reinforcement learning method is higher than baseline scenarios in which the gateway authenticates all the IoTDs with the same probability.
\end{itemize}

Simulation results show that, using the proposed {watermarking approach} and for a latency of under 1 second, the IoT signals can be reliably transmitted from IoTDs to the gateway. Moreover, in a massive IoT scenario, the learning algorithms improves the protection of the system by reducing by about $ 30\% $ the number of compromised IoTDs.

The rest of the paper is organized as follows. Section \ref{sect:sensor} introduces the system model for IoTD-gateway signal transmission in large-scale IoT systems and proposes the deep learning-based dynamic watermarking for IoTD-gateway signal authentication. Section \ref{sect:largescale} analyzes the authentication scenario in the case of a massive IoT with limited resources at the gateway. Section \ref{sect:sim} presents the simulation results and their analysis while conclusions are drawn in Section \ref{sect:conc}.

\section{IoT Signal Authentication: System Model and Deep Learning Solution}\label{sect:sensor}
Consider a massive IoT system having a set $\mathcal{N}$ of $ N $ IoTDs communicating with a gateway {such as a base station as shown in Fig. \ref{fig:system}\cite{Aazam2014}}. Any IoTD $ i $ in the system generates a signal $ y_i(t) $ at time step $ t $ with sampling frequency $ {f_i^s} $, and transmits this signal to the gateway which uses the received signal for estimation and control of the IoTD operation. In this system, we consider an adversary that seeks to compromise the IoTD by collecting the data from the communication link between the IoTD and the gateway, and, then, manipulates the transmitted signal. In this case, the transmitted signal from each IoTD $ i $ will be $ \bar{y}_i(t) \neq y_i(t) $ which will cause an estimation error at the gateway. Therefore, the gateway must implement an attack detection mechanism to differentiate between honest and false signals received from an IoTD. 

\begin{figure}[!t]
	\centering
	\includegraphics[width=\columnwidth]{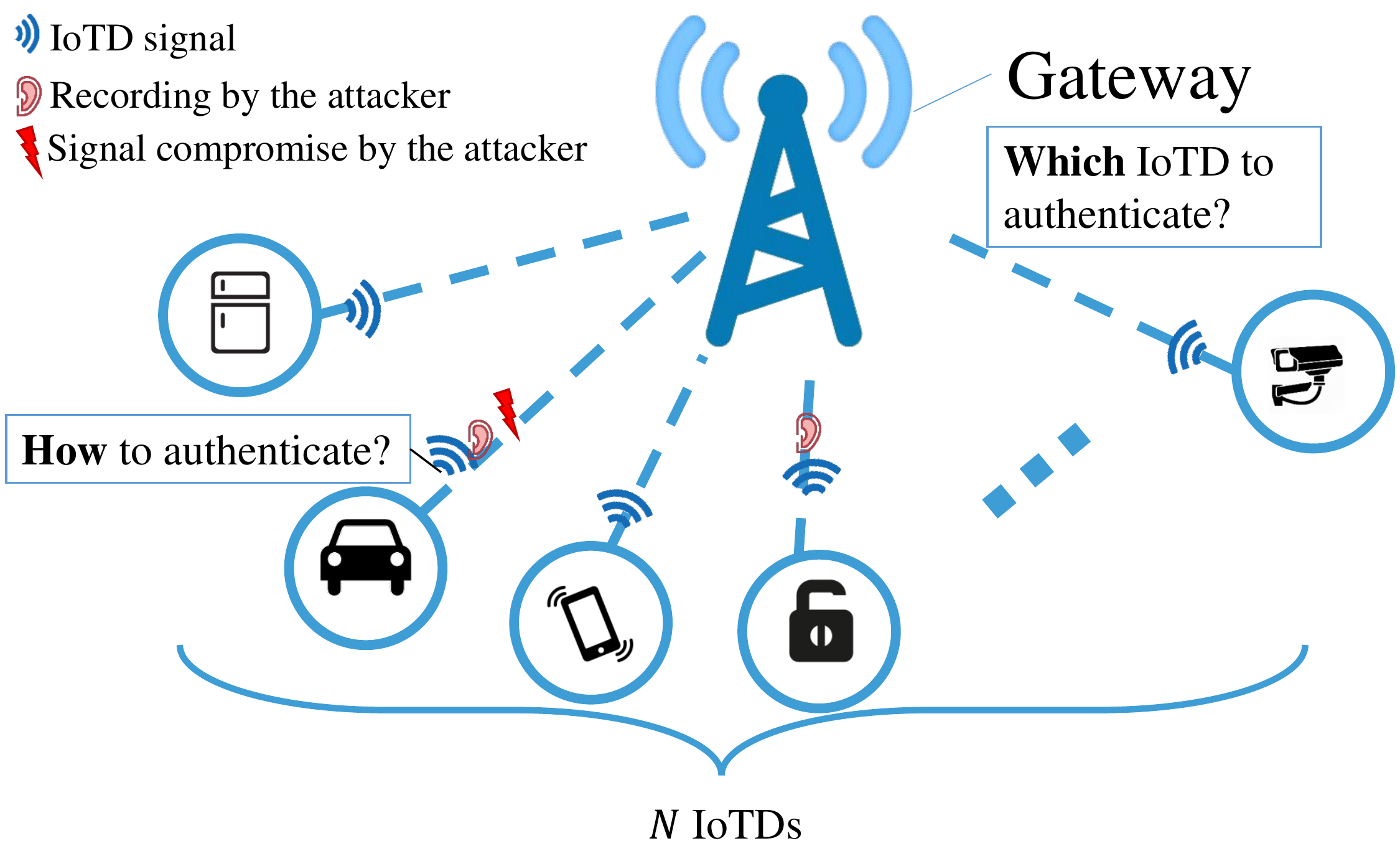}
	\captionsetup{justification=justified,singlelinecheck=false}
	\vspace{-4mm}
	\caption{\small
		 A massive IoT system with heterogeneous IoTDs transmitting their signals to a gateway. The gateway must solve two inter-related problems: 1) How to authenticate the IoTDs and 2) Which IoTDs to authenticate.}
	\label{fig:system}
	\vspace{-6mm}
\end{figure}

First, we analyze the signal authentication process between a single IoTD and the gateway and, then, we propose a dynamic watermarking framework based on deep learning to authenticate this signal communication with a very small delay. Subsequently, in Section \ref{sect:largescale}, we study the case in which, due to the massive nature of the system, the gateway is unable to authenticate all of the IoTDs and must choose an optimal subset to authenticate.

\subsection{Spread Spectrum Watermarking}
Watermarking uses a hidden, predefined non-perceptual code (bit stream) inside a multimedia signal to authenticate the ownership of such signals. One of the the most widely used watermarking methods is spread spectrum (SS) \cite{Improved2003Malvar} in which a key pseudo-noise sequence is added to the original signal. The watermarked signal from each IoTD $ i $ can then be written as follows:
\begin{align}\label{watermarking}
	w_i(t)=y_i(t)+\beta_i b p_i(t) \quad \textrm{for } t=1,\dots,n_i,
\end{align}
where $ w_i $ is the IoTD $ i $'s watermarked signal, $ p_i $ is a pseudo-noise binary sequence of $ +1 $ and $ -1 $ for IoTD $ i $, $ \beta_i $ is the IoTD $ i $'s relative power of the pseudo-noise signal to the original signal, $ b_i $ is the hidden bit in the signal which can take values $ +1 $ and $ -1 $, and $ n_i $ is the number of samples (frame length) of IoTD $ i $'s original signal used to hide a single bit. To extract the watermarked bit, the gateway receives the watermarked signal from each IoTD $ i $ and correlates it with the key pseudo-noise sequence. The extraction process will be:
	\begin{align}\label{correlation}
		\tilde{b}_i&=\frac{<w_i,p_i>_{n_i}}{\beta_i n_i}=\frac{<y_i,p_i>_{n_i}+\beta_i b_i<p_i,p_i>_{n_i}}{\beta_i n_i}=\tilde{y}_i+b_i,
	\end{align} 
where, for $ \tilde{b}_i>0 $, the extracted bit is $ 1 $, for $\tilde{b}_i<0 $, the extracted bit is $ -1 $ with $ <w_i,p_i>_{n_i} $ being the inner production of $ n_i $ samples of $ w_i $ and $ p_i $. $ p_i(t) $ and $ y_i(t) $ are independent stochastic variables at time $ t $. We assume that $ y_i(t) $ has mean $ \mu_i $ and variance $ \sigma_i^2 $. Next, we analyze the bit error rate of the extracted bit to evaluate the proposed watermarking scheme's performance.
\begin{lemma} \label{Theorem:BER}
	In the proposed SS watermarking scheme, the bit extraction error for IoTD $ i $ is $ \frac{1}{2}\textrm{erfc}\left(\frac{\beta_i\sqrt{n_i}}{\sigma_i \sqrt{2}}\right) $ {if $ n_{i_+} = n_{i_-} $, where $ n_{i_+} $ and $ n_{i_-} $ are the number of samples in $ p(t) $ with values equal to $ +1 $ and $ -1 $, respectively.}
\end{lemma}
\begin{proof}
	For $ \tilde{y}_i $, we can write:
	\begin{align}
		\tilde{y}_i&=\frac{1}{\beta_i n_i} \sum_{t=1}^{n} y_i(t)p_i(t)\nonumber=\frac{1}{\beta_i n_i}\left( \sum_{t\in \mathcal{P}_i^+} y_i(t)-\sum_{t\in \mathcal{P}_i^-} y_i(t)\right)\\\label{sum}
		&=\frac{n_{i_+}}{\beta_i n_i}\frac{1}{n_{i_+}}\sum_{t\in \mathcal{P}_i^+} y_i(t)-\frac{n_{i^-}}{\beta_i n_i}\frac{1}{n_{i_-}}\sum_{t\in \mathcal{P}_i^-} y_i(t),
	\end{align}
	where $ \mathcal{P}_i^+=\{t|p_i(t)=1\} $,  $ \mathcal{P}_i^-=\{t|p_i(t)=-1\} $. Since \eqref{sum} can be expressed as a sum of i.i.d variables, then, {for large values of $ n_{i_+} $ and $ n_{i_-} $} using the central limit theorem, we can write \eqref{sum} as a linear combination of two Gaussian distributions as $\tilde{y_i}=Y_{i_1}-Y_{i_2}, $ 
	where 
	\begin{align}
		Y_{i_1} \sim \mathcal{N}\left(\frac{{n_{i_+}}\mu_i}{\beta_i n_i},\frac{n_{i_+}}{\beta_i^2 n_i^2}\sigma_i^2\right),\,\,\,Y_{i_2} \sim \mathcal{N}\left(\frac{{n_{i_-}}\mu_i}{\beta_i n_i},\frac{n_{i_-}}{\beta_i^2 n_i^2}\sigma_i^2\right).
	\end{align}
	Since $ \tilde{y} $ is a linear combination of two independent Gaussian distributions, we have:
	\begin{align}
			\tilde{y}_i& \sim \mathcal{N} \left(\frac{{n_{i_+}}\mu_i}{\beta_i n_i}-\frac{{n_{i_-}}\mu_i}{\beta_i n_i},\frac{n_{i_+}}{\beta_i^2 n_i^2}\sigma^2+\frac{n_{i_-}}{\beta_i^2 n_i^2}\sigma_i^2\right),\nonumber\\
			\Rightarrow
		\tilde{y}_i&\sim \mathcal{N} \left(\frac{{(n_{i_+}-n_{i_-})}\mu_i}{\beta_i n_i},\frac{1}{\beta_i^2 n_i}\sigma_i^2\right).
	\end{align}
	Since $ n_{i_+}$ and $ n_{i_-} $ are design parameters, we choose $ n_{i_+}\hspace{-1mm}=n_{i_-} $. Thus, we have $
	\tilde{y}_i\hspace{-0.5mm}\sim\hspace{-0.5mm} \mathcal{N}\hspace{-0.5mm} \left(\hspace{-0.5mm}0,\frac{1}{\beta_i^2 n_i}\sigma_i^2\hspace{-0.5mm}\right)\hspace{-0.5mm}.$ Now, we can show that $ \tilde{b} $ is a Gaussian variable since it is a summation of a constant value with a Gaussian variable: 
	$
	\tilde{b_i} \sim \mathcal{N}\left(E(\tilde{b}_i)=b_i,\sigma^2_{\tilde{b}_i}=\frac{\sigma_i^2}{\beta_i^2 n_i}\right).
	$
	Then, to analyze the probability of error we consider $ b_i=1 $. In this case an error occurs when $ \tilde{b}_i<0 $, therefore, the probability of an error is:
	\begin{align}
	\textrm{Pr}\left\{\tilde{b}_i<0|b_i=1\right\}&=\frac{1}{2}\textrm{erfc}\left(\frac{E(\tilde{b}_i)}{\sigma_{\tilde{b}_i}\sqrt{2}}\right)=\frac{1}{2}\textrm{erfc}\left(\frac{\beta_i\sqrt{n_i}}{\sigma_i \sqrt{2}}\right).\label{errprob}
	\end{align}
	The same error probability can be obtained for $ b_i=-1 $. 
\end{proof}
From \eqref{errprob}, we can observe that, for large values of $ \beta_i $ and $ n_i $, the bit extraction error goes to zero. However, selecting large values for $ \beta_i $ and $ n_i $ will incur latency and computational costs for IoTDs and gateway, as discussed next.
\subsection{Static Watermarking for IoTD Attack Detection}
Now, using the SS method, we present a technique for authenticating the signals transmitted from an IoTD to the gateway. We first generate a random pseudo-noise binary sequence with $ n_i $ samples. Also, for every IoTD $ i $, we define a bit stream $ s_i $ with $ n_{s_i} $ samples. Then, using \eqref{watermarking}, we embed every bit of $ s_i $ in $ n_i $ samples of $ y_i $. Therefore, for any bit stream $ s_i $, we use $ n_in_{s_i} $ samples of $ y_i $, and this embedding procedure will repeat every $ n_in_{s_i} $ samples of $ y_i $. At the gateway, using \eqref{correlation}, we extract the bit stream. In case of a cyber attack, the received signal in the gateway will be $ \bar{y}_i(t) $ rather than $ w_i(t) $, and, hence, the extracted bit stream will differ from $ s_i $. Thus, the gateway will trigger an alarm for declaring the existence of a cyber attack. Fig. \ref{fig:static} shows the block diagram of static watermarking for attack detection at the level of an IoTD.
\begin{figure}[!t]
	\centering
	\includegraphics[width=\columnwidth]{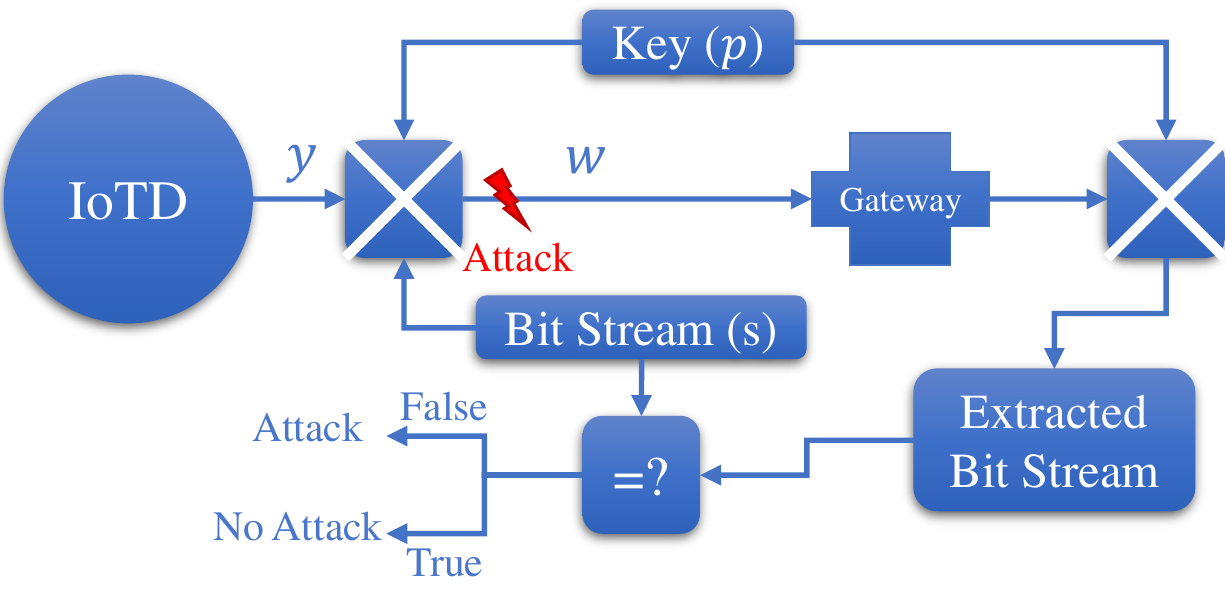}
	\vspace{-4mm}
	\caption{\small Static watermarking for attack detection.}
	\label{fig:static}
	\vspace{-6mm}
\end{figure}

In our proposed watermarking scheme, $ \beta_i $, $ n_i $, and $ n_{s_i} $ play crucial roles in the security of a given IoTD. The value of $ \beta_i $ must be much smaller than the value of $ \sigma_i $. Indeed, for comparable values of $ \beta_i $, an attacker can extract the key $ p_i $ and bit stream $ s_i $, since $ w_i(t)\simeq\beta_i b_i p_i(t) $ if $ |\beta_i b_i p(t)|= \beta_i >>|y_i(t)| $. Therefore, we must choose small values for $ \beta_i $. However, from Lemma \ref{Theorem:BER}, we know that a small $\beta_i$ will yields a higher bit error rate. To overcome this issue, we have to increase the pseudo-noise key length, $ n_i $. Although increasing the value of $ n_i $ will reduce the bit error rate, for large values of $ n_i $, the bit extraction procedure in \eqref{correlation} will result in higher computation load and will also cause higher latency since the gateway must wait for $ n_in_{s_i} $ samples from the IoTD to detect the attack. Moreover, large values of $ n_{s_i} $ will also cause larger delay at the gateway. Therefore, next, we propose a method to choose suitable values for these three parameters.

\begin{theorem}\label{Theorem:hyperparameters}
	To reduce the attacker's ability to extract the hidden bit stream as well as minimize the attack detection delay of the watermarking scheme $ \beta_i,n_i$, and $n_{s_i} $ must be selected to satisfy the following conditions:
	\begin{align}
		\frac{1}{2}\textrm{erfc}\left(\frac{(1+\frac{\mu_{1_i}}{\beta_i^2n_i^2})\beta^2_in_i\sqrt{n_i}}{ \sqrt{2(\sigma_{1_i}^2+2\sigma_i^2)}}\right)&\geq1-\underbar{P}, \label{eq1}\\
		\frac{1}{2}\textrm{erfc}\left(\frac{\beta_i\sqrt{n_i}}{\sigma_i \sqrt{2}}\right)&\leq \bar{\text{P}}, \label{eq2}\\
		n_{s_i} &\leq \frac{d{f_i^s}}{n_i} \label{eq3},
	\end{align}
	where $ d $ is the acceptable delay (in seconds) for attack detection, $ \mu_{1_i} $ and $ \sigma^2_{1_i} $ are the mean and variance of the multiplication of random variables $ y_{1_i}(t) $ and $ y_{2_i}(t) $ with the same distribution as $ y_i(t) $, $ \underbar{P} $ is our desired probability of unsuccessful attack, and $ \bar{\text{P}}$ is our desired bit extraction error probability. 
\end{theorem}
\begin{proof}
	See Appendix ~\ref{App:hyperparameters}.
\end{proof}
Using \eqref{eq1}, \eqref{eq2}, and \eqref{eq3}, we can find the values for the three parameters which satisfy our performance and delay constraints. For example, in Fig. \ref{fig:hyper}, we consider a scenario in which $ \mu_{1_i} = 0 $, $ \delta_{1_i}=0.5 $, $ \delta_i=0.5 $, $ \underbar{P}=0.05$, and $\bar{\text{P}}=0.01 $. We can see from Fig. \ref{fig:hyper} that, $ \beta_i=0.5 $ and $ n_i=10 $ can satisfy conditions \eqref{eq1} and \eqref{eq2}. Now, if we consider that $ f^s_i=1000 $, from condition \eqref{eq3}, to be able to have a delay at most $ d=0.1 $ seconds, we must have $ n_{s_i}=10 $. The proposed SS watermarking method can detect a cyber attack which can only change the transmitted data from an IoTD to the gateway. By choosing optimal values for the three watermarking parameters, the gateway can authenticate the transmitted data. Consider a case in which an attacker can also collect data from the IoTDs. Here, the attacker can launch a more complex {\emph{dynamic data injection attack} by recording the transmitted data from the IoTD and, then, summing the recorded data for a long period of time. Such an attack can potentially reveal the key $ p $.} For example, if the attacker collects the data for $ m $ windows of size $ nn_s $ and adds this data together, it obtains: $
\bar{w}_{m_i}(t)=\sum_{i=1}^m y_i^j(t) +m\beta_i b_i p_i(t),
$
where $ y_i^j $ is the signal received in window $ j$ from an IoTD, and $ \bar{w}_{m_i} $ is the summation of collected data. If we consider $ \sigma^2_{m_i} $ as the variance of the sum of $ m $ random variable $ y_i^1(t),\dots,y_i^m(t) $, then, there will exist a value for $ m $ where $ m^2\beta_i^2 >> \sigma_{m_i}^2 $. Hence, the attacker can use $ \bar{w}_{m_i} \simeq m\beta_i b_i p_i $ as the key for watermarked signal. This attack is successful because the embedded stream $ s $ is static, at all times. However, if the bit stream $ s_i $ changes dynamically in each window of $ n_in_{s_i} $ samples then, the system can deter such a dynamic data injection attack. Next, we propose a dynamic bit stream generation using deep learning.
\begin{figure*}[!t]
	\begin{subfigure}{0.5\textwidth}
		\centering
		\includegraphics[width=\columnwidth]{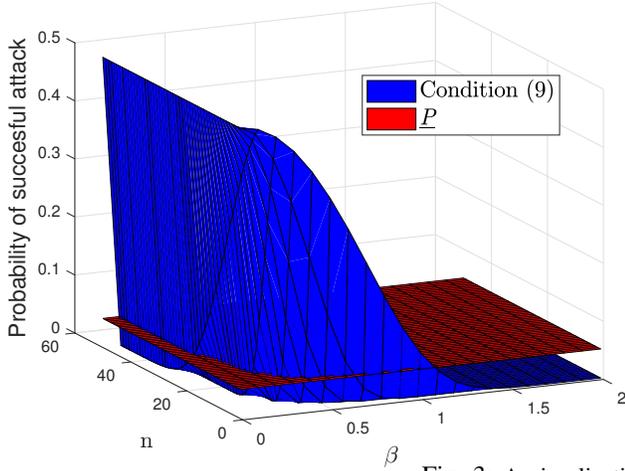}
	\end{subfigure}
	\begin{subfigure}{0.5\textwidth}
		\centering
		\includegraphics[width=\columnwidth]{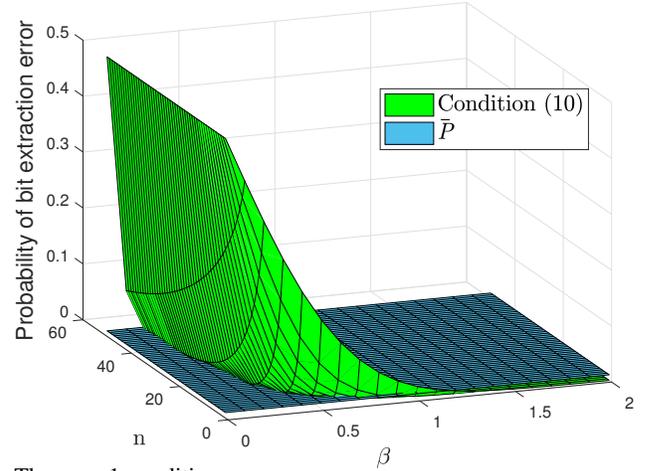}
	\end{subfigure}
	\vspace{-4mm}
	\caption{\small A visualization for Theorem \ref{Theorem:hyperparameters} conditions. }
	\label{fig:hyper}
	\vspace{-6mm}
\end{figure*}
\subsection{Deep Learning for Dynamic IoT Signal Watermarking}\label{sect:DeepMark}
To improve our authentication scheme, we propose a novel deep learning watermarking method for dynamically generating the bit stream $ s_i $ which can thwart {dynamic data injection} attacks. In the proposed dynamic watermarking scheme, we use the fingerprints of the signal $ y_i $ generated by an IoTD to dynamically update the bit stream, $ s_i $. Signal fingerprints are unique identifiers of a signal that can be mapped to a bit stream. Signal stochastic characteristics such as spectral flatness, central moments, skewness, and  kurtosis can be used for extraction of fingerprints from signals \cite{audio2002Haitsma,Wavelet2012Bertoncini,Wavelet1995Learned,Mellin2012Harley}. Due to time dependence of the IoTD signal stream on past time steps, we use the powerful deep LSTM framework, one of the most effective deep learning methods for sequence analysis, to extract the IoTD signal fingerprints\cite{chen2017machine,challita2017proactive}, and \cite{speech2013Graves}. 
\subsubsection{Introduction to LSTM Cells}
{
Recurrent neural networks (RNNs) are a special type of artificial neural networks (ANNs) which are useful for time series analysis \cite{chen2017machine}. ANNs are function approximators which map an input vector to a target vector. They are made of layers of artificial neurons which take a vector as an input, aggregate such input vector, and pass the aggregation through an \emph{activation function}. The output of the last layer's activation function can be considered as the target vector. While such ANNs map an input vector to an output vector, RNNs also feedback the output to their input which makes them suitable to analyze a time dependence in the input. Fig. \ref{fig:RNNs} shows a normal neuron and a recurrent neuron which are used in different ANN architectures. Each of these neurons will have a number of hyperparameters called \emph{bias} and \emph{weights} in the aggregation and activation steps. Thus, the aim of using this RNN is to derive those hyperparameters such that one can find an approximation of mapping from input to output. One of the widely used methods to train the RNN and derive the hyperparamters is the \emph{gradient descent} algorithm \cite{chen2017machine}, which uses the chain rule and gradient of the RNN's layers to find the optimal hyperparamters which minimize the difference between the RNN's output and the target vector. However, RNNs have a \emph{vanishing gradient} problem \cite{speech2013Graves} during training phase which means that the gradient of layers close to the input does not change much compared
to layers close to the output which makes the convergence of the gradient descent algorithm challenging. Moreover, conventional RNNs such as the ones shown in Fig. \ref{fig:RNNs} tend to forget the information about past samples and emphasize on learning from recent samples. To overcome these challenges, we use LSTM blocks \cite{speech2013Graves} which have three main components as shown in Fig. \ref{fig:LSTM}: 1) A forget gate which receives an extra input called the cell state input and learns how much it should memorize or forget from the past, 2) An input gate which aggregates the output of past steps and the current input and passes it through an activation function as done in a conventional RNN, and 3) An output gate which combines the current cell state and the output of input gate and generates the LSTM output. Next, we explain how we use the LSTMs for IoTD watermarking.

}

\begin{figure}[t]
	\centering
	\includegraphics[width=\columnwidth]{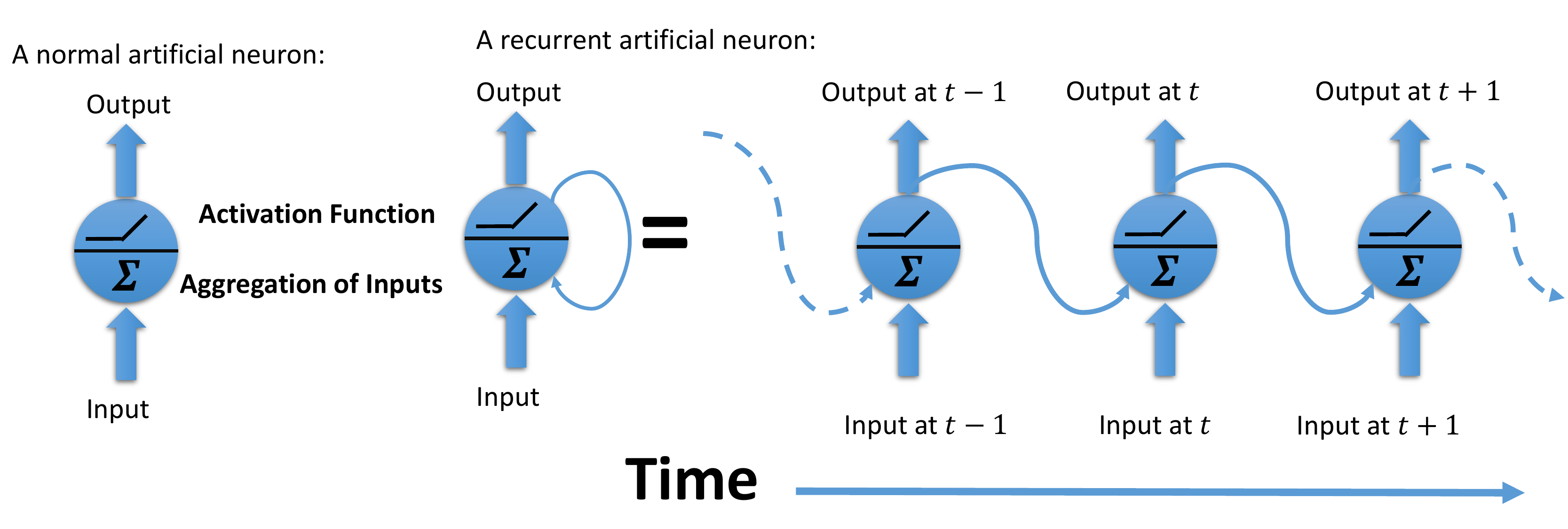}
	\vspace{-4mm}
	\caption{\small A block diagram of artificial neurons used in ANNs.}
	\label{fig:RNNs}
	\vspace{-2mm}
\end{figure}

\begin{figure}[t]
	\centering
	\includegraphics[width=\columnwidth]{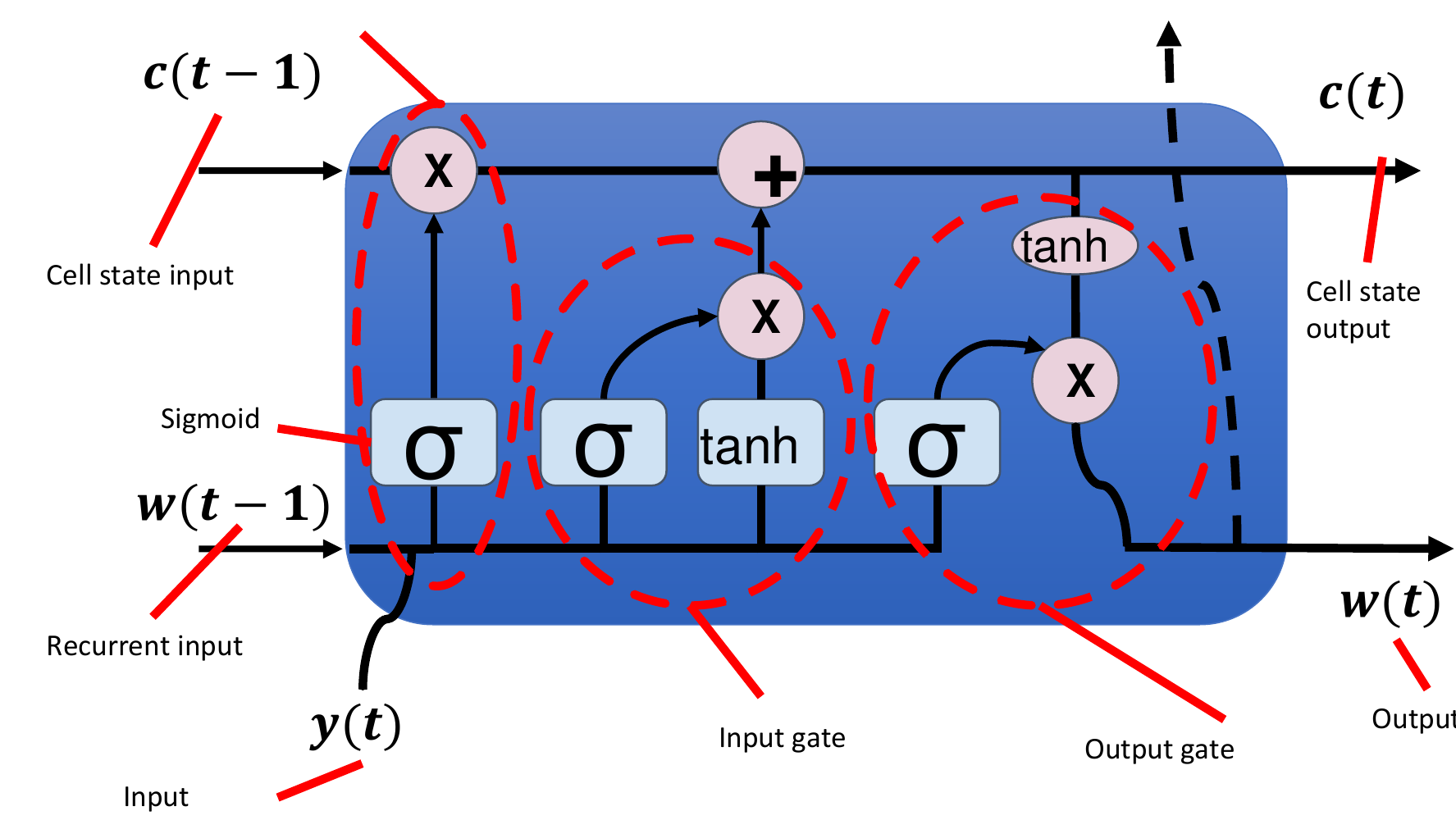}
	\vspace{-4mm}
	\caption{\small A generic LSTM block architecture.}
	\label{fig:LSTM}	
	\vspace{-4mm}
\end{figure}
\subsubsection{LSTM for Dynamic Signal Watermarking at the IoTDS}
To dynamically extract fingerprints from IoTD signals, we use an LSTM algorithm that allows an IoTD to update the bit stream based on the sequence of generated data. An LSTM algorithm processes an input $ \left(y_i(1),\dots,y_i(n_in_{s_i})\right) $ by adding new information into a memory, and using gates which control the extent to which new information should be memorized, old information should be forgotten, and current information should be used. The output of an LSTM algorithm will be impacted by the network activation at previous time steps and, hence, LSTMs are suitable for our IoT application in which we want to extract fingerprints from signals which are dependent on previous time steps.

\begin{figure*}[t]
	\begin{subfigure}{0.5\textwidth}
		\centering
		\includegraphics[width=0.65\columnwidth]{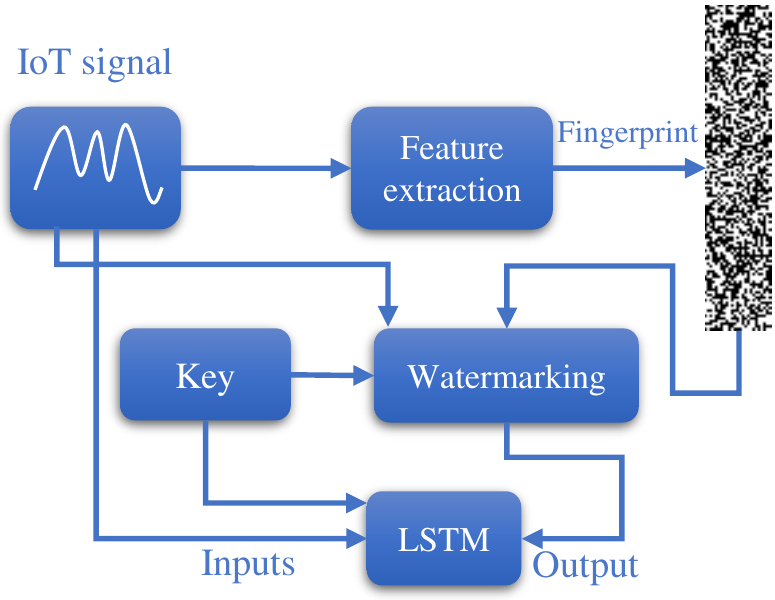}
		\caption{Training phase of an LSTM in an IoTD.}
		\label{fig:training}
	\end{subfigure}
	\begin{subfigure}{0.5\textwidth}
		\centering
		\includegraphics[width=\columnwidth]{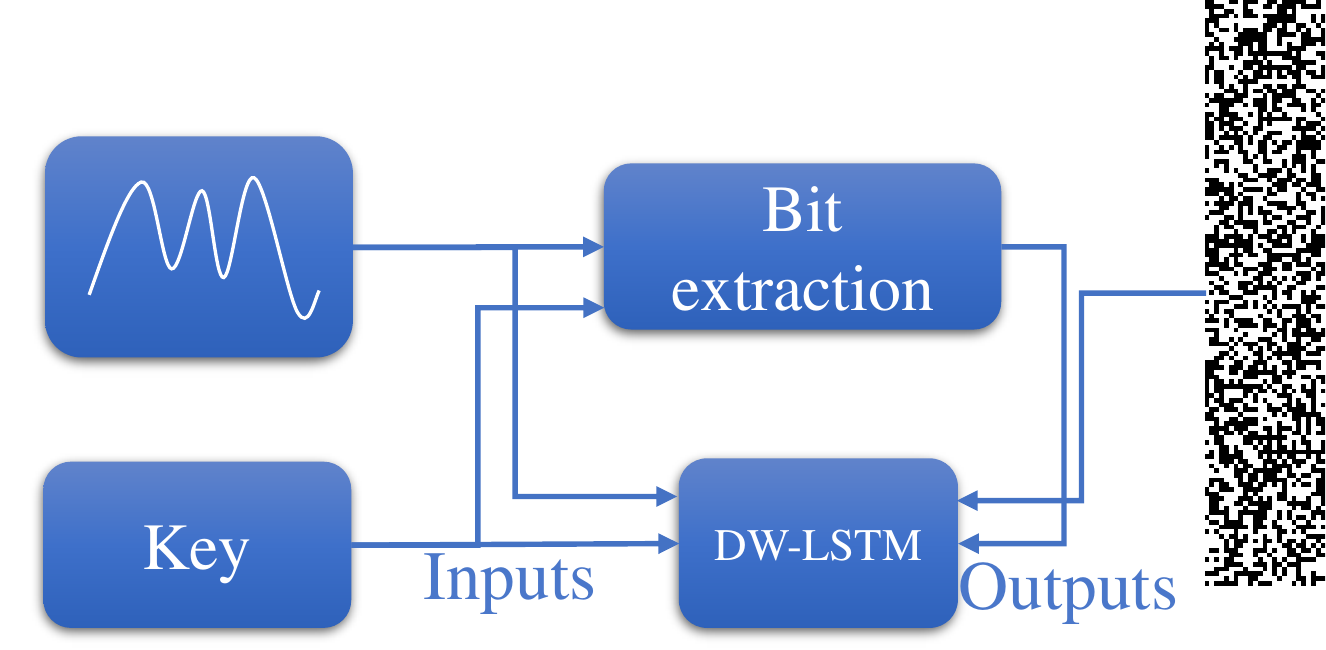}
		\caption{Training phase of an LSTM at the gateway.}
		\label{fig:trainingcpu}
	\end{subfigure}
	\caption{\small Training phase of the proposed deep learning methods.}
	\label{fig:train}
	\vspace{-6mm}
\end{figure*}
During the training phase, the parameters of the LSTM algorithm are learned from a given training dataset of different IoTDs such as accelerometer, gyroscope and positioning devices. As done in \cite{audio2002Haitsma,Wavelet2012Bertoncini,Wavelet1995Learned,Mellin2012Harley}, we choose spectral flatness, mean, variance, skewness, and kurtosis as features that are extracted from a signal of length $ n_in_{s_i} $ and then map these values to a bit stream with length $ n_i $. Next, we watermark this extracted bit stream into the original signal using a key. {To train the LSTM, we use the original signal $ y_i $ and a pseudo noise key $ p_i $ as input stream and the watermarked signal $ w_i $ as output stream.} Fig. \ref{fig:training} shows the training phase. Next, we illustrate how we use the trained LSTM to dynamically watermark an IoT signal. 

\subsubsection{LSTM for Dynamic Signal Authentication at the gateway}
At the gateway, we use a dynamic watermarking LSTM (DW-LSTM) for bit extraction. {To train this DW-LSTM, we use the watermarked signal $ w_i $ and key $ p_i $ as inputs to the neural network and the features of the original signal and extracted bit stream as outputs.} The block diagram model of the training phase at the gateway for our DW-LSTM is shown in Fig. \ref{fig:trainingcpu}. Using the DW-LSTM block of Fig. \ref{fig:training} at the IoTD and the DW-LSTM block of Fig. \ref{fig:trainingcpu} at the gateway, we propose a dynamic LSTM watermarking scheme to implement an attack detector at the gateway. 

In this method, a predefined bit stream is not used since a bit stream is dynamically generated inside the LSTM blocks at the IoTD and the gateway. This dynamic bit stream generation at the hidden layers of LSTMs solves the {dynamic data injection} attack problem, since recording and summing the IoTD signals will not increase the power ratio of the key sequence to the signal and the attacker will not be able to extract the key and bit stream. Using this method, the IoTD inserts the generated signal and key in its LSTM block in each window of $ n_in_{s_i} $ samples and produces a watermarked signal with different bit stream $ s_i $ in each window. At the gateway, the received watermarked signal and key are passed from the LSTM. Then, the two outputs (the extracted bits, and extracted features) are compared. In case of dissimilarity between two sequences, an attack alarm is triggered. Fig. \ref{fig:dynamicwatermarking} shows the block diagram if the proposed dynamic LSTM watermarking for attack detection. {In addition, note that the computationally expensive phase in this algorithm is actually the training phase where several iterations are done during gradient descent algorithm to find the optimal weights for the neural network \cite{chen2017machine}. However, the testing phase is actually a series of summations, multiplications, and activation functions which are not computationally expensive and can be run on any device \cite{chen2017machine}.	The training phase of our proposed deep learning dynamic watermarking can be done offline	(e.g. on high performance computers), as already adopted in the literature \cite{chen2017machine}. Hence, for every IoTD, we train an LSTM network offline and then implement at the corresponding IoTD. After implementation, the IoTDs will not need to train their associated LSTM network online and they	will operate in the testing phase.\footnote{The training phase can be done online as well by using some recent advances in LSTM architectures.}} Also, the proposed algorithm enables the gateway to authenticate any IoTD signal with a delay of $ d $. To analyze the required computational resource at the gateway, next, we derive the computational complexity of the proposed watermarking algorithm. 

\begin{proposition}\label{Proposition:complexity}
	The complexity of the proposed signal authentication method at the gateway is bounded by $ \mathcal{O}(d{f_i^s}) $.
\end{proposition}
\begin{proof}
	From \eqref{correlation}, we know that, in order to extract each bit in stream $ s $, we need $ n $ multiplications and $ n-1 $ summations. Therefore, the complexity of extracting one bit is $ \mathcal{O}(n_i) $ and extracting all of the bits in $ s $ will have a complexity of $ \mathcal{O}(n_in_{s_i}) $. Moreover, from Theorem \ref{Theorem:hyperparameters}, we know that $ n_i $ is chosen based on the watermarking performance criteria, while $ n_{s_i} $ is bounded by $ \frac{d{f_i^s}}{n_i} $. Thus, the complexity of signal authentication at the gateway is bounded by $ \mathcal{O}\left(\frac{d{f_i^s}}{n_i}n_i\right)= \mathcal{O}\left({d{f_i^s}}\right) $.
\end{proof}
\begin{figure}
	\centering
	\includegraphics[width=\columnwidth]{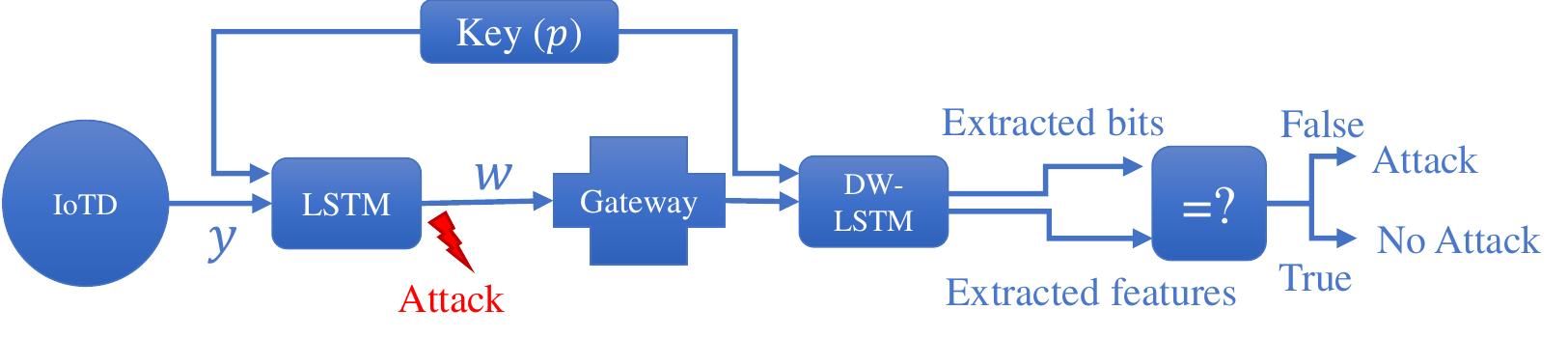}
	\vspace{-4mm}
	\caption{\small Dynamic watermarking for attack detection.}
	\label{fig:dynamicwatermarking}
	\vspace{-6mm}
\end{figure}
Proposition \ref{Proposition:complexity} shows that the sampling rate of an IoTD's signal directly affects the complexity of the authentication process. However, as previously mentioned, an IoTD with higher sampling and packet transmission rates to the gateway will be more valuable for the IoT system\cite{2017Zhou}. Thus, under limited computational resources, \emph{the gateway can only authenticate a limited number of IoTD signals} using our proposed method and, thus, it has to choose which subset of $\mathcal{N}$ it can authenticate. This computational limitation provides an opportunity for the attacker to choose the IoTDs with unauthenticated signals to attack and stay undetected. We assume that the amount of computations that can be done at the gateway is bounded by $\mathcal{O}(C) $ which means that the total authentication complexity of all the received signals cannot exceed $ \mathcal{O}(C) $:
\begin{align}\label{eq:complexity}
\mathcal{O}\left(d\sum_{i\in\mathcal{S}}{f_i^s}\right)\leq \mathcal{O}(C),
\end{align}
where $ \mathcal{S} \subseteq \mathcal{N} $ is the set of IoTDs whose signals will be authenticated by the gateway and $ {{f_i^s}} $ is the sampling rate of IoTD $ i $. Since the arguments on both sides of \eqref{eq:complexity} are linear, we have:
\begin{align}\label{eq:resources1}
\sum_{i\in\mathcal{S}}{f_i^s}\leq \frac{C}{d}.
\end{align}
Since the sampling rate $ f_s $ of each IoTD corresponds to its value for the gateway, $ v_i $, by considering a value proportional\footnote{Without loss of generality, our approach can accommodate any other relationship between the IoTD's value and sampling frequency.} to each IoTD's sampling frequency, $ v_i=\frac{{f_i^s}}{\sum_{i=1}^N{f_i^s}} $, we can rewrite \eqref{eq:resources1} as:
\begin{align}\label{eq:resources}
\sum_{i\in\mathcal{S}}{v_i}\leq \frac{C}{d}\sum_{i=1}^N{f_i^s}\triangleq R,
\end{align}
 Thus, the gateway must choose a set $\mathcal{S} \subseteq \mathcal{N}$ that satisfies \eqref{eq:resources}. In addition, the attacker has a limitation on the number of IoTDs that it can attack simultaneously due to its limited available resources. We capture this resource limitation by assuming that the attacker has only $ K $ devices that it can use to {dynamically inject data to IoTD signals}. Thus, the attacker must choose $ K $ target IoTDs, that are not $ \mathcal{S} $, while the gateway must predict the IoTDs that the attacker will target and, then, it will authenticate them. Since the attacker's and the gateway's actions are \emph{interdependent}, the outcome of their decisions requires analyzing their interaction. In the following, we address this interaction between the gateway and the attacker in the large-scale IoT system using a game-theoretic approach\cite{bacsar1998dynamic}.
\section{Game Theory for Authentication under Computational Constraints}\label{sect:largescale}
In the considered massive IoT system, we assume that IoTDs that have more data to send will be more valuable for the system since more important applications require more frequent monitoring and control \cite{2017Zhou}. However, IoTDs with more valuable data are also more likely to be selected as an attack target by the adversary. As discussed in Section \ref{sect:sensor}, IoTD signal authentication using our proposed technique requires computational resources at the gateway to process the received data from all IoTDs. Therefore, the gateway must optimally predict the vulnerable IoTDs while the attacker must predict which unauthenticated IoTDs to target so as to maximize the disruption in the IoT system. We analyze this problem using game theory. 
\subsection{Game Formulation}
To model the interdependent decision making processes of the attacker and the gateway, we introduce a noncooperative game $ \left\{\mathcal{P},\mathcal{Q}^j,u^j,K,C\right\} $ defined by five components: a) the \emph{players} which are the attacker $ a $ and the gateway $ {g} $ in the set $ \mathcal{P}\triangleq \left\{a,g\right\} $, b) the \emph{strategy} spaces $ \mathcal{Q}^j $ for each player $ j\in \mathcal{P} $, c) a \emph{utility function}, $ u^j $ for each player, d) the number of attacker's devices, $ K $, which can record IoT signals and e) the available computational resources for the gateway, $ C $. {Note that the gateway can get the information about $ K $ using the typical attackers' capabilities, past attacks, or known data on	similar attacks. Moreover, the attacker can extract information about the computational capability of gateways from data sheets and documentations, particularly for standard IoT systems that typically deploy known types of hardware and software.} For the gateway, the set of \emph{pure} strategies $ \mathcal{Q}^{g} $ corresponds to different feasible IoTD subsets whose signals can be authenticated without exceeding the available computational resources:
\begin{align}
	\mathcal{Q}^{g}=\left\{\mathcal{S}\subset\mathcal{N}\Big|\sum_{i\in\mathcal{S}}v_i\leq R\right\}.
\end{align}
On the other hand, for the attacker the set of pure strategies $ \mathcal{Q}^a $ is a set of $K$ IoTDs that will be targeted by an attack: $
	\mathcal{Q}^a=\left\{\mathcal{K}\subset\mathcal{N}\Big||\mathcal{K}|\leq K \right\}.
$
Moreover, the utility function\footnote{{Here, we use $ v_i=\frac{{f_i^s}}{\sum_{i=1}^N{f_i^s}} $ and the constraint in \eqref{eq:resources} between $ R $ and $ v_i $. The following game model and calculations hold true for any other relationship between $ R $ and IoTD values, $ v_i $. }} of each player can be written as follows:
\begin{equation}
\begin{aligned}
	u^{g}\left(\mathcal{S},\mathcal{K}\right)&=
	\sum_{i=1}^{N}v_i-\sum_{i\in\mathcal{K},i\notin\mathcal{S}}v_i=1-\sum_{i\in\mathcal{K},i\notin\mathcal{S}}v_i,\\
	u^a\left(\mathcal{K},\mathcal{S}\right)&=\sum_{i\in\mathcal{K},i\notin\mathcal{S}}v_i,
\end{aligned}
\end{equation}
where $ v_i $ is IoTD $ i $'s value. These utility functions essentially capture the fraction of secured IoTD signals for the defender, and the fraction of IoTDs that the attacker can compromise while remaining undetected. Therefore, the attacker seeks to maximize the fraction of compromised IoTDs while the defender seeks to minimize it. This coupling in the players strategies and utilities naturally leads to a game-theoretic situation \cite{Bacci2016}.
One of the most important solution concepts for noncooperative games is that of a \emph{Nash equilibrium} (NE). The NE characterizes a state at which no player $ j $ can improve its utility by changing its own strategy, given the strategy of the other player is fixed. For a noncooperative game, the NE in pure (deterministic) strategies can be defined as follows:
\begin{definition}
	A \emph{pure-strategy Nash equilibrium} of a noncooperative game is a vector of strategies $ \left[\mathcal{X}^{g^*},\mathcal{X}^{a^*}\right] \in \mathcal{Q}^{g}\times\mathcal{Q}^{a} $ such that $ \forall j \in \mathcal{P} $, the following holds true:
	$
		u^j\left(\mathcal{X}^{j^*},\mathcal{X}^{-j^*}\right)\geq u^j\left(\mathcal{X}^j,\mathcal{X}^{-j^*}\right), \,\, \forall \mathcal{X}^j\in\mathcal{Q}^j,
	$ 
	where $ -j $ is the identifier for $ j $'s opponent.
\end{definition}
The NE characterizes a stable game state at which the gateway cannot improve the protection of IoTD signals by \emph{unilaterally} changing its action $ \mathcal{S} $ given that the action of the attacker is fixed. Moreover, at the NE, the attacker cannot manipulate more IoTD signals by changing its action $ \mathcal{K} $ while the gateway keeps its action $ \mathcal{S} $ fixed. Before analyzing the NE of our game, we introduce the useful concept of a \emph{dominated strategy} that we will use in our further analysis.
\begin{definition}
	Player $ j $'s strategy $ \mathcal{X}^j $ is \emph{weakly dominated} if there exists another strategy $ \tilde{\mathcal{X}}^j\subset \mathcal{Q}^j $ such that:
	$
		u^j(\mathcal{X}^j,\mathcal{X}^{-j})\leq u^{j}(\tilde{\mathcal{X}}^j, \mathcal{X}^{-j}),\,\, \forall \, {\mathcal{X}^{-}\in \mathcal{Q}^{-j}},
	$
	with the strict inequality for at least one $  \mathcal{X}^{-j} $.
	In case the above inequality holds strictly for all $  \mathcal{X}^{-j}\in\mathcal{Q}^{-j} $, $  \mathcal{X}^j $ is said to be \emph{strictly dominated} (by $ \tilde{ \mathcal{X}}^j $).
\end{definition}
Essentially, a strategy is dominated if choosing it always yields a smaller utility compared to any other strategy, given all possible strategies for other players. Using this definition, prior to finding the NE, we next derive the dominated strategies of the gateway and the attacker.
\begin{proposition}\label{Proposition:domgateway}
	The gateway must choose as many IoTDs as possible to authenticate and, thus, any strategy $ \mathcal{S} \in \mathcal{Q}^{g} $ is weakly dominated by $ \tilde{\mathcal{S}}  \in \mathcal{Q}^{g} $ if $ \mathcal{S}\subseteq \tilde{\mathcal{S}} $.
\end{proposition}
\begin{proof}
	From set theory we have:
	\begin{align}
		\mathcal{S}\subseteq \tilde{\mathcal{S}},
	&\Rightarrow	\mathcal{K}\cap\mathcal{S}\subseteq \mathcal{K}\cap \tilde{\mathcal{S}},
		\Rightarrow	\mathcal{K}-\left(\mathcal{K}\cap \tilde{\mathcal{S}}\right)\subseteq \mathcal{K}-\left(\mathcal{K}\cap {\mathcal{S}}\right),\nonumber\\
		&\Rightarrow	\left\{i\Big|i\in\mathcal{K},i\notin\tilde{\mathcal{S}}\right\}\subseteq \left\{i\Big|i\in\mathcal{K},i\notin\mathcal{S}\right\}.
	\end{align}
	Therefore, we have:
	\begin{align}
		\sum_{i\in\mathcal{K},i\notin\tilde{\mathcal{S}}}v_i&\leq \sum_{i\in\mathcal{K},i\notin\mathcal{S}}v_i,\nonumber\\
		u^{g}\left(\mathcal{S},\mathcal{K}\right)=
		1-\sum_{i\in\mathcal{K},i\notin\mathcal{S}}v_i&\leq	1-\sum_{i\in\mathcal{K},i\notin\tilde{\mathcal{S}}}v_i=u^{g}\left(\tilde{\mathcal{S}},\mathcal{K}\right),
	\end{align}
	which proves that $ \mathcal{S} $ is weakly dominated by $ \tilde{\mathcal{S}} $.
\end{proof}
Proposition \ref{Proposition:domgateway} shows that the gateway must use all of its available computational resource to authenticate the IoTD signals since any strategy that is a subset of another strategy uses less computational resources.
\begin{proposition}\label{Proposition:domAtt}
	The attacker must use all of its $ K $ recording devices, i.e., any strategy $ \mathcal{K} \in \mathcal{Q}^a $ is weakly dominated by $ \tilde{\mathcal{K}}  \in \mathcal{Q}^a $ if $ \mathcal{K}\subseteq \tilde{\mathcal{K}} $ and $ \left|\tilde{\mathcal{K}}\right|=K $.
\end{proposition}
\begin{proof}
	We know that  $ \mathcal{K}\subseteq \tilde{\mathcal{K}} $, and, thus, the number of IoTDs that are attacked by the adversary when choosing strategy $ \tilde{\mathcal{K}} $ is higher than the number of attacked IoTDs when choosing strategy $ \mathcal{K} $. Therefore, if the gateway chooses any strategy $ \mathcal{S} $, then the number of unauthenticated IoTD devices for the first case will be greater than or equal to the latter case or:
	\begin{align}
		\left\{i\Big|i\in{\mathcal{K}},i\notin{\mathcal{S}}\right\}&\subseteq \left\{i\Big|i\in\tilde{\mathcal{K}},i\notin\mathcal{S}\right\},
	\end{align}
	and, thus, we have:
	\begin{align}
		u^a\left(\mathcal{K},\mathcal{S}\right)=
		\sum_{i\in\mathcal{K},i\notin\mathcal{S}}v_i&\leq	\sum_{i\in\tilde{\mathcal{K}},i\notin{\mathcal{S}}}v_i=u^a\left(\tilde{\mathcal{K}},\mathcal{S}\right).
	\end{align}
	Since the attacker cannot choose more than $ K $ IoTDs to attack, thus any strategy $ \tilde{\mathcal{K}} $ with $ K $ members, $ \left|\tilde{\mathcal{K}}\right|=K $, dominates all the strategies that are $ \mathcal{K}\subseteq\tilde{\mathcal{K}} $.
\end{proof}

Proposition \ref{Proposition:domAtt} shows that non-dominated strategies for the attacker are those that include $ K $ IoTDs to attack. Here, we define $ \tilde{\mathcal{Q}}^j\subseteq \mathcal{Q}^j $ as the set of player $ j $'s non-dominated strategies. The attacker's non-dominated strategies can be interpreted as combination of $ K $ IoTDs from all $ N $ IoTDs. Therefore, the number attacker's strategies is $ {{N}\choose{K}} $. Since the number of attacker's recording devices, $ K $, is comparably less than the number of IoTDs $ N $, we can easily see that the complexity of finding the attacker's non-dominated strategy set is $ \mathcal{O}\left(N^K\right) $ which is comparably smaller than considering all the subsets of $ \mathcal{N} $, that results in a complexity of $ \mathcal{O}(2^N) $. Thus, Proposition \ref{Proposition:domAtt} reduces the complexity of the attacker's game significantly.

\begin{algorithm}[t]
	\caption{Dynamic Programming for finding the gateway's Non-dominated Strategies}
	\begin{algorithmic}[1]\footnotesize 
		\State \textbf{Input} $ \mathcal{F}\triangleq\left\{{f^s_1},\dots,{f^s_N}\right\} $, $ C/d $,
		\State \textbf{Initialize} $\boldsymbol{M}_{N+1\times C/d-{f^s_1}+1}$ everywhere False apart from $ \boldsymbol{M}[0,0]= $ True, 
		\State Start filling all the entities of matrix $ \boldsymbol{ M} $:
		\State \textbf{for} $i\leftarrow 1$ \emph{to} $N$ \textbf{do}
		\State \quad  \textbf{for} $j\leftarrow {f^s_1}$ \emph{to} $C/d$ \textbf{do}
		\State \quad \quad $\boldsymbol{ M}[i,j]=\boldsymbol{ M}[i-1,j]\bigvee \boldsymbol{ M}[i-1,j-{f_i^s}], $ (True, if there is a subset with the values less than $ {f_i^s} $ that sum up to $ j $.)
		\State \textbf{for} $k\leftarrow C/d$ \emph{to} ${f^s_1}$ \textbf{do}
		\State \quad \textbf{if} $ \max_{m}\boldsymbol{ M}[m,k-{f_i^s}] $ \textbf{do} (Checks the IoTD with highest $ {f^s} $ such that the IoTDs with less $ {f^s} $ can sum up to $ k $.)
		\State \quad \quad \textbf{Initialize} $ \mathcal{R}=\left\{\right\} ,$
		\State \quad \quad $ \mathcal{S}\leftarrow $ \textbf{RecPath}$ (m,k,\mathcal{F},\boldsymbol{ M},\mathcal{R}) $: (This function follows the path in $ \boldsymbol{ M} $ until reaching to the first column.)
		\State \quad \quad \quad \quad \quad \textbf{if} $ \boldsymbol{ M}[m,k] $ \textbf{do}
		\State \quad \quad \quad \quad \quad \quad $ \mathcal{B}\leftarrow \mathcal{R}, $
		\State \quad \quad \quad \quad \quad \quad  \textbf{RecPath}$ (m-1,k,\mathcal{F},\boldsymbol{ M},\mathcal{B}) $
		\State \quad \quad \quad \quad \quad \textbf{if} $ k \geq {f^s}_m \bigwedge \boldsymbol{ M}[m,k-{f^s}_m] $ \textbf{do}
		\State \quad \quad \quad \quad \quad \quad $ \mathcal{R}\leftarrow {f_i^s} ,$
		\State \quad \quad \quad \quad \quad \quad  \textbf{RecPath}$ (m-1,k,\mathcal{F},\boldsymbol{ M},\mathcal{R}) $
		\State \quad \quad \quad \quad \textbf{Output} $ \mathcal{R} $
		\State \textbf{Output} $ \mathcal{S}$
	\end{algorithmic}
	\label{Algorithm:SubsetSumProblem}
\end{algorithm}
For the gateway, finding all the strategies that satisfy the condition in Proposition \ref{Proposition:domgateway} is an NP-hard problem \cite{Pisinger2003}. Thus, in Algorithm~\ref{Algorithm:SubsetSumProblem} we propose a novel dynamic programing approach to reduce the complexity of finding the non-dominated strategies of the gateway. The algorithm takes the set of all IoTD sampling frequencies and $ C/d $ as input. In this algorithm, we first define a $( N+1 )\times (C/d-f^s_1+1) $, matrix $ \boldsymbol{ M} $ whose element $ \boldsymbol{ M}[i,j] $ is set to a ``True'' value if there is a subset of IoTDs with a sampling frequency less than $ {f_i^s} $ such that the summation of all IoTD sampling frequencies in this subset equals $ j $. 
\begin{figure}[!t]
	\centering
	\includegraphics[width=\columnwidth]{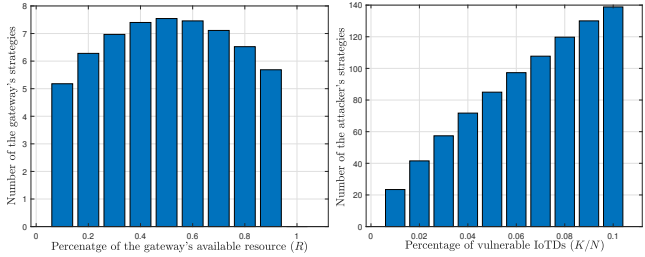}
	\caption{\small Number of the gateway's and the attacker's strategies for 1000 IoTDs (logarithmic scale).}
	\label{fig:strategies}
	\vspace{-1mm}
\end{figure}
\begin{corollary}
	At each stage of computation, the Algorithm \ref{Algorithm:SubsetSumProblem} uses the solutions of previous subproblems and since the operations used to fill each entity of $ \boldsymbol{ M} $ are similar, the complexity of finding the defender's non-dominated strategies reduces from $ \mathcal{O}(2^N) $ to $ \mathcal{O}(N C/d) $.
\end{corollary} 
Even though Algorithm ~\ref{Algorithm:SubsetSumProblem} reduces the complexity of finding the gateway's non-dominated strategies to a linear time, the number of these non-dominated strategies is dependent on the gateway's available resources and the distribution of sampling frequencies. Fig. \ref{fig:strategies} shows the number of non-dominated strategies for the gateway and the attacker, when $ N=1000 $ IoTDs, $ R\in(0,1]$, and $ K\in[0.01N,0.1N] $. In this massive IoT scenario, from Fig. \ref{fig:strategies}, we can see that, the number of strategies for both players is very large which consequently requires a complex process to find the pure-strategy NE. Moreover, even though we derived the non-dominated strategies for both players, the NE is not guaranteed to exist for our game\cite{bacsar1998dynamic}. For example, consider only three IoTDs with $ \left\{1000,2000,3000\right\} $ as their sampling frequency, assume that $ C/d=5000 $, and an attacker having 1 recording device. This example game, along with the non-dominated strategies of both players, is summarized in Table \ref{Table:game}. Any element $ (i,j) $ in Table \ref{Table:game} is the outcome of playing the gateway's $ i $-th and the attacker's $ j $-th strategy. From Table \ref{Table:game}, we observe that, for any outcome of the game, at least one of the players can change its strategy to gain a better payoff. Therefore, this game cannot have a pure-strategy NE for a general case. Thus, we investigate the NE in \emph{mixed strategies} which is guaranteed to exist for finite noncooperative games\cite{bacsar1998dynamic}. When using mixed strategies, each player will assign a probability for playing each one of its pure strategies. For a massive IoT, the use of mixed strategies is motivated by two facts: a) the gateway and the attacker must randomize over their strategies in order to make it nontrivial for the opponent to guess their potential action, and b) the procedure of choosing IoTDs can be repeated over an infinite time duration and mixed strategies can capture the frequency of choosing certain strategies for both players. Thus, next, we analyze our game's mixed-strategy NE.

\begin{table}[!t]
	\centering
	\begin{tabular}{ |c|c|c|c| } 
		\hline
		\diagbox{\textbf{Attacker}}{\textbf{gateway}}& $ \left\{1000,2000\right\} $ & $ \left\{1000,3000\right\} $ & $ \left\{2000,3000\right\} $  \\
		\hline
		$ \left\{1000\right\} $ & $ (0,1) $ & $ (0,1) $ & $ (1/6,5/6) $ \\ 
		\hline
		$ \left\{2000\right\} $ & $ (0,1) $ & $ (2/6,4/6) $ & $ (0,1) $ \\
		\hline
		$ \left\{3000\right\} $ & $ (3/6,3/6) $ & $ (0,1) $ & $ (0,1) $\\
		\hline
	\end{tabular}
	\caption{\small An example of strategies and utilities for the game between the gateway and the attacker.}
	\label{Table:game}
	\vspace{-6mm}
\end{table}
\subsection{Mixed-Strategy Nash Equilibrium}
In our game, by using mixed strategies, the attacker and defender will assign probabilities for playing each one of their non-dominated strategies\cite{bacsar1998dynamic}. Let $ \boldsymbol{p}^a $ be the vector of mixed strategies for the attacker where each element in $ \boldsymbol{p}^a $ is the probability of choosing a set of IoTDs, i.e., selecting one strategy from the attacker's strategy set $ \mathcal{K} $. Moreover, $ \boldsymbol{p}^{g} $ is the vector of mixed strategies for the gateway whose elements represent the probability of choosing a certain strategy from the gateway's strategy set, $ \mathcal{S} $. Consequently, each player must choose its own mixed-strategy to maximize its expected utility which is defined by:
\begin{align}
	U^j(\boldsymbol{p}^j,\boldsymbol{p}^{-j})=\sum_{\mathcal{S}\in\mathcal{Q}^{g}}\sum_{\mathcal{K}\in\mathcal{Q}^a}\boldsymbol{p}^{g}(\mathcal{S})\boldsymbol{p}^a(\mathcal{K})u^j(\mathcal{S},\mathcal{K}), \,\,\,\text{for }\forall j \in \mathcal{P}.
\end{align}
To solve this problem, we seek to find the mixed-strategy Nash equilibrium, defined as follows:
\begin{definition}
	A mixed strategy profile $ \boldsymbol{p}^* $ constitutes a \emph{mixed-strategy Nash equilibrium} (MSNE) if, for each player, $ j $, we have:
 $
		U^j(\boldsymbol{p}^{j^*},\boldsymbol{p}^{{-j}^*})\geq U^j(\boldsymbol{p}^{j},\boldsymbol{p}^{-j^*}), \,\,\, \forall \boldsymbol{p}^j\in \mathcal{P}^j,
$
	where $ \mathcal{P}^j $ is the set of all probability distributions for player $ j $ over its action space $ \mathcal{Q}^j $.
\end{definition}
The MSNE for our game implies a state at which the gateway has chosen its optimal randomization over authenticating the signals of its IoTDs and, therefore, cannot further improve the system security by changing this randomization. Similarly, for the attacker, an MSNE is a state at which the attacker has chosen its probability distribution over the selection of IoTDs that it will attack and, thus, cannot improve its expected utility by changing its choice. Since our game is a constant-sum two-player game, the von Neuman indifference principle can be used to find a closed-form solution for the MSNE \cite{bacsar1998dynamic}. Under this principle, at the MSNE, the expected utilities of the players with respect to the mixed strategies played by the opponent  must be equal, for every pure strategy choice. To derive the MSNE for our game, we first define an \emph{allocation} vector $ \hat{\boldsymbol{p}}^j_{N\times 1} $ for each player $ j $ such that each element $ i $ in this vector is the probability of choosing IoTD $ i $. The relationship between the allocation vector and the mixed-strategy of our game can be written as follows:
\begin{align}\label{eq:probs}
\hat{p}^{g}_i=\sum_{\mathcal{S}\in\mathcal{S}_i}p^{g}(\mathcal{S}),\,\, \hat{p}^a_i=\sum_{\mathcal{K}\in\mathcal{K}_i}p^a(\mathcal{K}),
\end{align}
where $ \hat{p}^j_i $ is element $i$ of vector  $ \hat{\boldsymbol{p}}^j$. We define $ \mathcal{K}_i=\left\{\mathcal{K}\in\tilde{ \mathcal{Q}}^a|i\in\mathcal{K}\right\} $ as the set of all attacker strategies that have IoTD $ i $, and $ \mathcal{S}_i=\left\{\mathcal{S}\in\tilde{\mathcal{Q}}^{g}|i\in\mathcal{S}\right\} $ as the set of all gateway strategies that have IoTD $ i $. We  next derive the mapping between the expected utility of each player by playing mixed-strategy vector $ \boldsymbol{p}^j $ and the allocation vector $ \hat{\boldsymbol{p}}^j $, then we prove that our game has infinitely many, MSNEs all of which achieve a unique expected value for both attacker and defender.
\begin{proposition}\label{Proposotion:mapping}
	To map the mixed strategy vectors to allocation vectors the following conditions must hold true:
	\begin{align}\label{eq:mapping}
		 \sum_{i=1}^{N}\hat{p}_i^a=K,\,\,\, \sum_{i=1}^{N}\alpha_i\hat{p}_i^{g}=D,
	\end{align}
	where $ D $ is the maximum number of IoTDs in a strategy $ \mathcal{S} \in \tilde{ \mathcal{Q}}^{g} $ and, $\forall i \in \mathcal{N}$, $ \alpha_i $ is an integer.
\end{proposition}
\begin{proof}
	First, we analyze the mapping between attacker's allocation vector and mixed-strategy, $ \hat{p}_i $. From the definition of $ \mathcal{K}_i $, we have $ \sum_{\mathcal{K}\in\mathcal{K}_i}p^a(\mathcal{K})=\hat{p}^a_i$. Moreover, we know that the summation of all the attacker's non-dominated mixed strategies equals to $ 1 $, i.e., $ \sum_{\mathcal{K}\in\tilde{\mathcal{Q}}^a}p^a(\mathcal{K})=1 $. In addition, since every strategy $ \mathcal{K}\in \hat{\mathcal{Q}}^a $ has $ K $ IoTDs, then, if we make a set by $ K $ times repeating the set $ \tilde{\mathcal{Q}}^a $, we can build each $ \mathcal{K}_i $ from this set to obtain:
	\begin{align}
		\sum_{i=1}^{N}\hat{p}_i^a=K\sum_{\mathcal{K}\in\tilde{\mathcal{Q}}^a}p^a(\mathcal{K})=K.
	\end{align} 
	For the gateway, the procedure is similar. From the definition of $ \mathcal{S}_i $, we have  $ \sum_{\mathcal{S}\in\mathcal{S}_i}p^{g}(\mathcal{S})=\hat{p}^{g}_i$ and $\sum_{\mathcal{S}\in\tilde{\mathcal{Q}}^{g}}p^{g}(\mathcal{S})=1 $. Here, the number of IoTDs in each $ \mathcal{S}\in\tilde{ \mathcal{Q}}^{g} $ is not equal, however, to build $ \mathcal{S}_i, \forall i \in \mathcal{N} $, we must define a set with $ D $ repetitions of $ \tilde{ \mathcal{Q}}^{g} $ where $ D $ is the maximum number of IoTDs in a strategy $ \mathcal{S} \in \tilde{ \mathcal{Q}}^{g} $, i.e., $ D\triangleq\max_{\mathcal{S}\in \tilde{ \mathcal{Q}}^{g}}|\mathcal{S}| $. Since some of the IoTDs might not be included in the strategies which consist of $ D $ IoTDs, the set $ \mathcal{S}_i $ can repeat more than once for these IoTDs, and, therefore, we will have: $
	\sum_{i=1}^{N}\alpha_i\hat{p}_i^{g}=D\sum_{\mathcal{S}\in\tilde{\mathcal{Q}}^{g}}p^{g}({\mathcal{S}})=D,
	$
	where $ \alpha_i $ is the number of times $ \mathcal{S}_i $ shows up in $ D $ repetitions of set $ \tilde{ \mathcal{Q}}^{g} $.
\end{proof}
Proposition \ref{Proposotion:mapping} uncovers a linear relationship between the allocation probabilities. Using this relationship, and given that the attacker can have a successful attack on IoTD $ i $ if the defender does not authenticate IoTD $ i $, we can define the expected utility of the players as follows:
\begin{align}\label{eq:expected}
U^a(\hat{\boldsymbol{p}}^a,\hat{\boldsymbol{p}}^{g})&=\sum_{i=1}^{N}\hat{p}^a_i(1-\hat{p}^{g}_i)v_i, \, U^{g}(\hat{\boldsymbol{p}}^{g},\hat{\boldsymbol{p}}^a)=1-U^{g}(\hat{\boldsymbol{p}}^a,\hat{\boldsymbol{p}}^{g}),
\end{align}
with the condition in \eqref{eq:mapping}. Next, we derive the MSNE using \eqref{eq:mapping} and \eqref{eq:expected}.
\begin{theorem}\label{Theorem:MSNE}
	The defined game between the attacker and the gateway has infinitely many MSNEs that achieve a unique expected utility, $ V^{j^*} $, for each player $ j $. These MSNEs can be derived by solving the following minimax problem:
	\begin{equation}\label{eq:minimax}
		\begin{aligned}
		&V^{j^*}\triangleq 1-V^{-j^*}=\min_{\hat{\boldsymbol{p}}^{-j}}\max_{\hat{\boldsymbol{p}}^{j}}U^j(\hat{\boldsymbol{p}}^j,\hat{\boldsymbol{p}}^{-j}),\\ &\textrm{ s.t. } \,\, \sum_{i=1}^{N}\hat{p}_i^a=K,\,\,\, \sum_{i=1}^{N}\alpha_i\hat{p}_i^{g}=D.
		\end{aligned}
	\end{equation}
\end{theorem}
\begin{proof}
	Since the game between the players is constant-sum, then the expected utility of the game at MSNE is the solution of minimax problem in \eqref{eq:minimax}\cite{bacsar1998dynamic}. Moreover, the defined expected utilities and the constraints in \eqref{eq:mapping} and \eqref{eq:expected} are linear functions, therefore, the minimax problem in \eqref{eq:minimax} has a single solution which we call $ \hat{\boldsymbol{p}}^{j^*} $ and thus the expected utility is unique, $ V^{j^*} $. Moreover, using the mapping between the allocation vectors and mixed-strategy vectors we can find the mixed-strategies at MSNE, by solving a set of $ N $ equations in \eqref{eq:probs}. However, since the number of the attacker's and the gateway's strategies are greater than $ N $, then solving this set of equations will result in an infinite number of solutions, i.e., infinitely many MSNEs. 
\end{proof}
To solve problem \eqref{eq:minimax}, one must find all the non-dominated strategies of the gateway to derive values of $ D $ and $ \alpha_i $. However, as discussed before, finding all the gateway's non-dominated strategies is challenging in massive IoT scenarios. Therefore, in a massive IoT scenario analytically deriving the MSNE by using traditional algorithms such as minimax is computationally expensive. Moreover, the gateway will need to store the player's massive strategy set and re-run the entire steps of the conventional algorithms to reach an MSNE as the IoT system changes or a new IoTD joins to the system. Hence, the delay during the convergence process of such algorithms may not be tolerable for massive IoT scenarios. Also, at each time step, since the gateway cannot have complete information about the unauthenticated IoTDs due to its resource limitation which makes the convergence of conventional algorithms not suitable for finding MSNE. Therefore, we propose two learning algorithms: a) a \emph{fictitious play (FP)} for a complete information game where the gateway knows all IoTDs' states at each time step and b) a \emph{deep reinforcement learning (DRL)} algorithm that considers the gateway's lack of information about the unauthenticated IoTDs. 
\subsection{Fictitious Play for Complete Information}
To find the allocation vector at MSNE, $ \hat{\boldsymbol{p}}^{j^*}\hspace{-1mm} $, we propose a learning algorithm based on fictitious play (FP){\cite{heinrich2015fictitious}}. Since our two-player game is constant-sum, using the results of \cite{bacsar1998dynamic} and \cite{heinrich2015fictitious}, we can guarantee convergence FP to an MSNE. In the proposed algorithm, each player uses its belief about the allocation vector that its opponent will adopt. This belief stems from previous observations and is updated at every iteration. Given that the FP algorithm learns the allocation vector rather than the mixed strategies, it does not require storing the set of players' strategies thus significantly reducing the complexity of finding MSNE compared to von Neuman's approach. { Note that the von Neuman approach has a combinatorial complexity $ \mathcal{O}(2^N) $. In contrast, by using FP, the complexity reduces to $ \mathcal{O}(N) $}. Let $ \boldsymbol{\delta}^j(t) $ be player $ j $'s perception of the mixed strategy that $ -j $ adopts at time instant $ t $. Each element $\delta_i^j(t)$ of $ \boldsymbol{\delta}^j(t) $ represents the belief that $ j $ has at time $ t $, that is the probability of attacking IoTD $ i $ (for the attacker) or authenticating an IoTD $ i $ (for the defender). Such a belief can be built based on the empirical frequency with which $ j $ has chosen IoTD $ i $ in the past.  Thus, let $ \eta_i^j(t) $ be the number of times that $ j $ has observed $ -j $ choosing IoTD $  i $ up to time instant $ t $. Then, $ {\delta}^j_i(t) $ can be calculated as follows:
\begin{align}
	{\delta}^j_i(t)=\frac{\eta_i^j(t)}{\sum_{i=1}^{n}\eta_i^j(t)}.
\end{align}

To this end, at time instant $ t+1 $, based on the vector of empirical probabilities, $ \boldsymbol{\delta}^j(t) $, that it has perceived until time $ t $, each player $ j $ chooses the IoTDs that maximize its expected utility with respect to its belief about its opponent while considering both players' constraints, that is, the attacker chooses a set $ \mathcal{H}^a $ of $ K $ IoTDs such that:
\begin{align}\label{eq:brattacker}
	 \mathcal{H}^a(t)  = \arg \max_{\mathcal{K}}{U^a(\mathcal{K},\boldsymbol{\delta}^a(t) )},\,\, \text{ s.t }  |\mathcal{K}|=K.
\end{align}
Meanwhile, the gateway chooses a set of IoTDs such that the resulting computational complexity does not exceed its computational resource constraints, as follows:
\begin{align}\label{eq:brgateway}
\mathcal{H}^{g}(t)  = \arg \max_{\mathcal{S}}{U^{g}(\mathcal{S},\boldsymbol{\delta}^{g}(t) )},\,\, \text{ s.t }  \sum_{i\in\mathcal{H}^{g}(t)}v_i \leq R .
\end{align}
After each player $ j $ chooses its strategy at time instant $ t+1 $, it can update its belief as follows:
\begin{align} \label{eq:beliefupdate}
	\delta_i^j(t+1)=\frac{t}{t+1}\delta_i^j(t)+\frac{1}{t+1}\mathds{1}_{\left\{i\in \mathcal{H}^{-j}(t)\right\}}.
\end{align}
This learning process proceeds until the calculated empirical frequencies converge. Convergence is achieved when: $
	\left|\delta_i^j(t+1)-\delta_i^j(t)\right|=\epsilon, \,\,\, \forall i\in\mathcal{N}, \, \forall j \in \mathcal{P},
$ where $ \epsilon $ is the convergence error which is a very small number.

This algorithm assumes complete information for the players since they update their strategies based on \eqref{eq:brattacker} and \eqref{eq:brgateway} at each step, i.e., this algorithm assumes that the attacker knows which IoTDs the gateway authenticated at the previous step, and the gateway knows which IoTDs the attacker attacked at the previous step. This can be practical in scenarios in which even the unauthenticated IoTDs can report their security status to the gateway. However, in more realistic scenarios, the gateway can only have information about the IoTDs that it has authenticated at the previous steps. {Note that, we consider that the attacker can always observe the outcome of the game since any intelligent attacker will normally have an end-goal from its attack. For example, by attacking a smart grid sensor (which can be considered as an IoTD), the attacker can aim at changing the state estimation of the power grid \cite{watermark2015Mo} and \cite{Sanjab2016}}. Thus, next, we propose a DRL algorithm based on LSTM blocks \cite{heinrich2016deep} to consider the incomplete information case.

\begin{figure}[!t]
	\centering
	\captionsetup{justification=justified ,singlelinecheck=false}
	\includegraphics[width=\columnwidth]{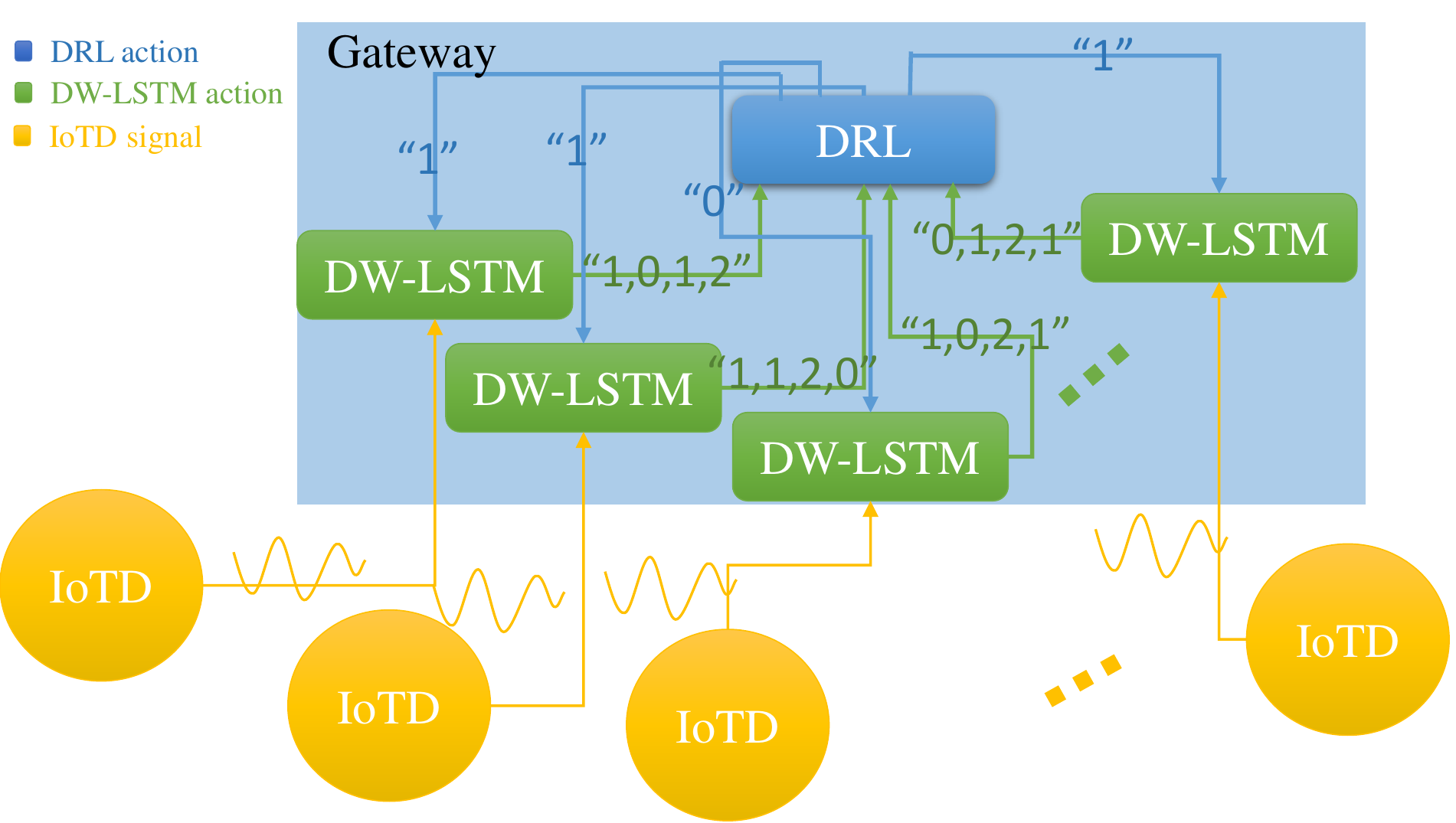}
	\caption{\small Security architecture for massive IoT scenario with incomplete information. At the gateway, DW-LSTM blocks receive the signal from IoTDs, and, then, they authenticate the received signal if DRL triggers them by sending a value $ 1 $. Finally, DW-LSTMs send their associated IoTD's state to the DRL.}
	\label{fig:massiveiot}
	\vspace{-6mm}
\end{figure}
\vspace{-5mm}
\subsection{Deep Reinforcement Learning for Authentication under Incomplete Information}
To capture the incomplete information at the gateway, we propose a DRL{\cite{heinrich2016deep}} method using LSTM blocks as shown in Fig. \ref{fig:massiveiot}. The architecture consists of two components: a) a deep learning-based dynamic watermarking for IoTD signal authentication as discussed in Section \ref{sect:sensor} and b) a DRL algorithm based on LSTM blocks to learn which IoTDs to authenticate at each step based on the attacker's previous actions and the gateway's computational constraints.

{As discussed for the FP algorithm, the players use \eqref{eq:brattacker} and \eqref{eq:brgateway} to find the best action at each step and update their belief about their opponent's actions. The attacker can observe its immediate payoff resulting from the action that it takes at each time step (since the attacker can see if its attack was successful or not) and then choose its best strategy at each time step using \eqref{eq:brattacker}. However, the gateway cannot compute its immediate payoff at each time step since the gateway can only observe the security state of the authenticated IoTDs, and, hence, the state of the unauthenticated IoTDs will remain unknown to the gateway. Thus, the gateway cannot find its best action at each time step using \eqref{eq:brgateway}. To overcome this challenge, we propose to use deep neural networks to approximate the the gateway's payoff at each time step. In particular, we use an LSTM at the gateway which receives the action stream of the attacker at the previous time steps, approximates the gateway's payoff at each time step and, then, chooses the gateway's optimal action. As discussed previously, LSTM blocks are very useful in sequence to sequence mapping and future prediction based on past sequences. Such a (deep reinforcement learning) DRL algorithm has been used in many applications such as in imperfect-information games \cite{heinrich2016deep}, autonomous vehicles \cite{fridman2018deeptraffic}, and human-level control systems \cite{mnih2015human}.
	
The proposed DRL algorithm has two components: (i) a deep neural network (DNN) that summarizes the past actions of the attacker and (ii) a reinforcement learning (RL) component, which can be used by each player to decide on the best action to choose based on the summary from the DNN, as shown in Fig. \ref{fig:DeepRL}.
	
To derive the gateway's action that maximizes its expected utility, we use a Q-learning algorithm (a special RL method) \cite{heinrich2016deep}. In this algorithm, we define a state-action value $ Q $-function, $ Q(\mathcal{S},\boldsymbol{h}) $, which is the expected return of the gateway when starting at a \emph{state} $ \boldsymbol{h} $ and performing action $ \mathcal{S} $. The state $ \boldsymbol{h} $ is an $ N $ dimensional attacker action sequence which is defined as follows: Each dimension $ i $ of $ \boldsymbol{ h} $ is a stream $ {\boldsymbol{ h}_i(t)=} \left[h_i^1(t),\dots,h_i^q(t)\right] $ from each DW-LSTM at the gateway where $ h_i^l(t) $ is the action of the attacker on IoTD $ i $ at step $ t-(q-l+1) $. This action can take any value in $ \left\{0,1,2\right\} $ where $ 0 $ indicates no attack, $ 1 $ indicates under attack, and $ 2 $ indicates that the gateway did not authenticate this IoTD at this time instant. To derive the maximizer action at each time step for the gateway, we use the following update rule for the $ Q $-function (known as Bellman equation\cite{heinrich2016deep}):
\begin{align}\label{eq:Qfunction}
Q_{t+1}(\mathcal{S}(t),\boldsymbol{h}(t))\hspace{-0.5mm}=Q_t(\mathcal{S}(t),\boldsymbol{h}(t))+\alpha \Big[U^g(\mathcal{S}(t),\boldsymbol{\delta}^g(t))\nonumber\\
+\gamma\max_{\mathcal{S}}\hspace{-0.5mm}Q_{t+1}(\hspace{-0.5mm}\mathcal{S}\hspace{-0.5mm},\hspace{-0.5mm}\boldsymbol{h}(t+1)\hspace{-0.5mm})-\hspace{-0.5mm}Q_{t}\hspace{-0.3mm}(\hspace{-0.3mm}\hspace{-0.5mm}\mathcal{S}(t)\hspace{-0.5mm},\hspace{-0.5mm}\boldsymbol{h}(t)\hspace{-0.5mm})\hspace{-0.5mm}\Big],
\end{align}
where $ \alpha $ is a learning rate and $ \gamma $ is a discount factor. From \eqref{eq:Qfunction}, we can see that, at each step, the gateway must find the action which maximizes  $ Q_t $. However since each $ h^l_i(t) $ can take any value in $ \{0,1,2\} $ and $ \boldsymbol{h} $ has $ {q\times N} $ elements, thus, $ \boldsymbol{h} $ can have $ 3^{q\times N} $ different values. Storing this many values in a table and finding the maximizer action is challenging.  
	
To solve such a challenging problem, we propose to use LSTM blocks that can store information for long periods of time and, thus, can learn long-term dependencies within a given sequence. Thus, the proposed DRL algorithm will use a DNN as shown in Fig. \ref{fig:DeepRL} to approximate the $ Q $ function for the gateway and using this $ Q $ function we will choose optimal actions for the gateway from \eqref{eq:Qfunction}. This then allows to overcome the complexity challenges. Algorithm \ref{Algorithm:DeepRL} summarizes the proposed DRL approach that is used by the gateway to learn its optimal action vectors. This algorithm takes the sampling frequency and the gateway's resource as input, then, starts learning from a random belief for each player. {The gateway uses its belief vector at each time step to estimate the state of the unauthenticated IoTDs. To this end, every IoTD that is not authenticated at time step $ t $, i.e., $\forall l\in\left\{1,\dots, N\right\}: h_l(t) = 2 $, is assigned a value $ h_l(t) =1 $ with probability $ \delta_l^g(t) $ and $ h_l(t) =0 $ with probability $ 1-\delta_l^g(t) $. Then the gateway can approximate its payoff.} While the attacker performs the fictitious play as in the previous section, the gateway aims at first approximating its payoff (note that the gateway cannot find its exact payoff since it does not have complete information about the IoTD states) and then finding the best action that maximizes its payoff at each time step. Moreover, Fig. \ref{fig:DeepRL} shows the DNN architecture for the proposed DRL algorithm. Note that in this architecture the output of the LSTM block is fed to a fully connected layer which is followed by a regression layer. In a fully connected layer, neurons have connections to all activations in the previous layer, as seen in regular neural networks. In DNNs, high-level reasoning is done via fully connected layers\cite{chen2017machine}.}
\begin{algorithm}[t]
	\caption{Deep Reinforcement Learning for Authentication Under Incomplete Information.}
	\begin{algorithmic}[1]\footnotesize 
		\State \textbf{Input} $ \mathcal{F}\triangleq\left\{v_1,\dots,v_N\right\} $, $ R $, 
		\State Initialize one \emph{replay memory} $ M $ that stores the past experiences of the gateway, one DNN for $ Q $, $\boldsymbol{\delta}^a(0)$, and $\boldsymbol{\delta}^{g}(0)$.  
		\State \textbf{Repeat:}
		\State \quad Select an action $ \mathcal{K} $ for the gateway using \eqref{eq:brattacker},
		\State \quad Observe state $ \boldsymbol{ h}(t) $ for the gateway.
		\State \quad \quad {$ \forall i \in \left\{1,\dots,N\right\} $ if $ h_i(t) = 2 $, then $ h_i(t) = 1 $ with probability $ \delta^g_i(t) $, otherwise  $ h_i(t) = 0 $.}
		\State \quad \quad with probability $ \varepsilon $ select a random action $ \mathcal{S} \in \tilde{ \mathcal{Q}}^g $,
		\State \quad \quad otherwise select $ \mathcal{S} = \argmax_{\mathcal{S}'} Q(\mathcal{S}',\boldsymbol{h}(t)) $ s.t. $ \mathcal{S}' \in \tilde{ \mathcal{Q}}^g $.
		\State \quad Perform actions $ \mathcal{S} $ and $ \mathcal{K} $ for both players simultaneously. 
		\State \quad Observe utility $ U^{g}(\mathcal{S},\boldsymbol{\delta}^{g}(t+1)) $ and new state $ \boldsymbol{h}(t+1) $.
		\State \quad Store \emph{experience} $ \left\{\boldsymbol{h}(t),\mathcal{S}(n),U^{g}(\mathcal{S}(t+1),\boldsymbol{\delta}^{g}(t+1)),\boldsymbol{h}(t+1)\right\} $ in replay memory $ M $ for the gateway.
		\State \quad Sample a random experience $ \left\{\hspace{-0.5mm}\hat{\boldsymbol{h}}(\tau), \hspace{-0.5mm}\hat{\mathcal{S}}(\tau),\hspace{-0.5mm}\hat{U}^{g}(\hat{\mathcal{S}}(t+1),\hat{\boldsymbol{\delta}}^{g}(t+1)),\hspace{-0.5mm}\hat{\boldsymbol{ h}}(\tau+1)\hspace{-0.5mm}\right\} $ from the replay memory $ M $.
		\State \quad Calculate the \emph{target} value $ t $ for the gateway:
		\State \quad \quad If the sampled experience is for $ \tau=0 $ then $ t=\hat{U}^{g}  $,
		\State \quad \quad Otherwise $ t=\hat{U}^{g} + \gamma \max_{\mathcal{S}'} Q(\mathcal{S},\alpha'^j) $.
		\State \quad Train the network $ Q^j $ for each player using $ [t^j-Q(\hat{\mathcal{S}}(\tau),\hat{\boldsymbol{ h}}(\tau) ]^2$.
		\State \quad $ t = t+1 $.
		\State \quad Update beliefs $ \boldsymbol{\delta}^{g}(t) $ and $ \boldsymbol{\delta}^a(t) $ using \eqref{eq:beliefupdate}.
		\State \textbf{Until} convergence.
	\end{algorithmic}
	\label{Algorithm:DeepRL}
\end{algorithm}
\begin{figure*}
	\centering
	\includegraphics[width=0.8\textwidth]{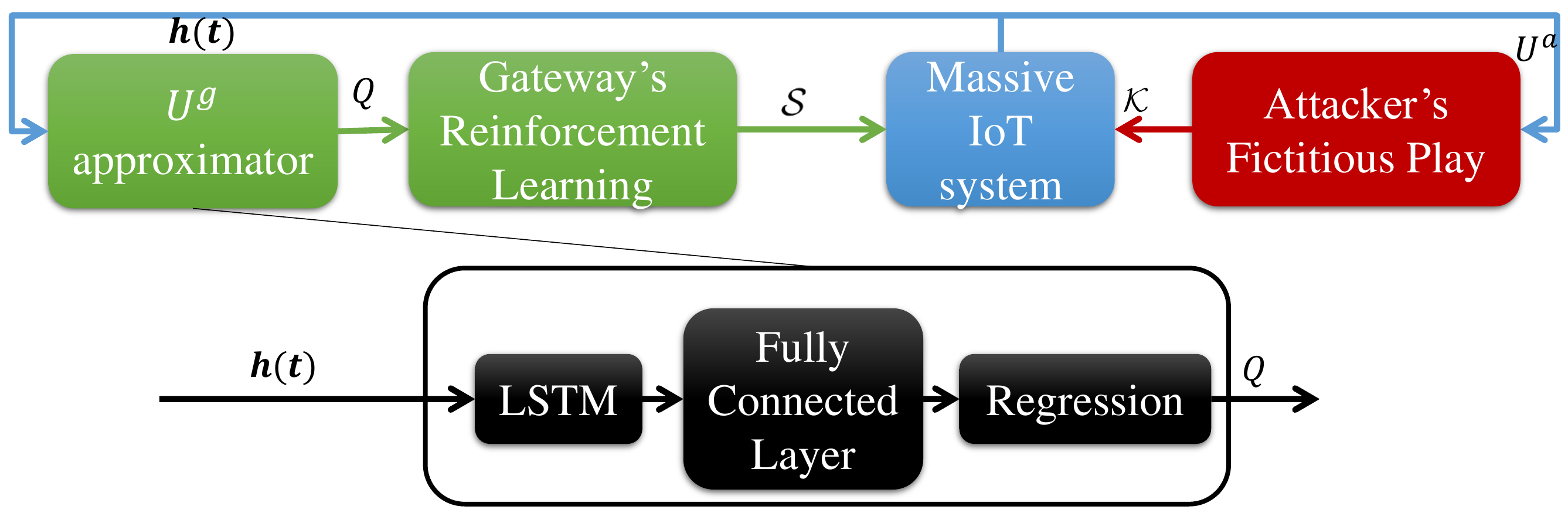}
	\caption{\small The proposed DRL architecture.}
	\label{fig:DeepRL}
	\vspace{-4mm}
\end{figure*}

Although this algorithm might not converge to an MSNE, since it does not have complete information about all the IoTDs' state, however, since the gateway uses an LSTM method to predict the attacker's future actions based on the interdependence of attacker's past actions, it can choose a set of IoTDs at each step that minimizes the number of compromised IoTDs. This, in turn, will yield a desirable solution for the system under incomplete information. In Section \ref{sect:sim}, the simulation results will show that the expected utility resulting from this DRL algorithm is higher than baseline scenarios, in which the gateway authenticates IoTDs with equal probabilities or proportional to IoTDs' values.
\section{Simulation Results and Analysis}\label{sect:sim}
For our simulations, we use a real dataset from an accelerometer with sampling frequency $ f_s=1 \text{kHz} $. In each simulation, we derive the optimal values for $ \beta $, $ n $, and $ n_s $ using the method proposed in Section \ref{sect:sensor} such that they satisfy the reliability and delay constraints.

\begin{figure*}
	\begin{subfigure}{0.5\textwidth}
		\centering
		\includegraphics[width=\columnwidth]{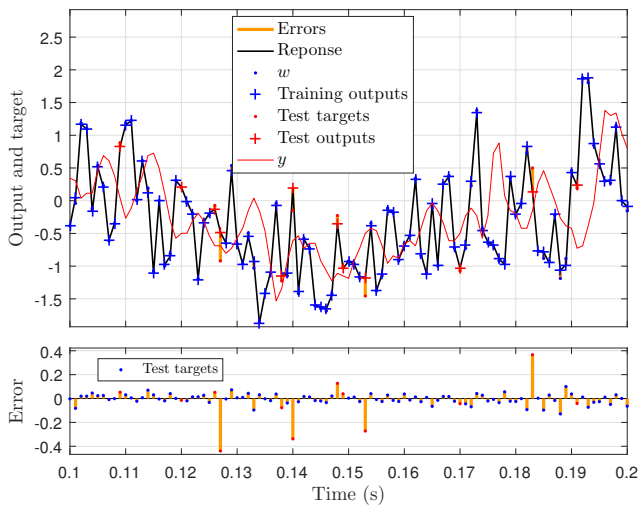}
		\caption{\small Comparison of IoT signal, static watermarked signal, dynamic watermarked signal, and training error.}
		\label{fig:trainingout}
	\end{subfigure}
	\begin{subfigure}{0.5\textwidth}
		\centering
		\includegraphics[width=0.9\columnwidth]{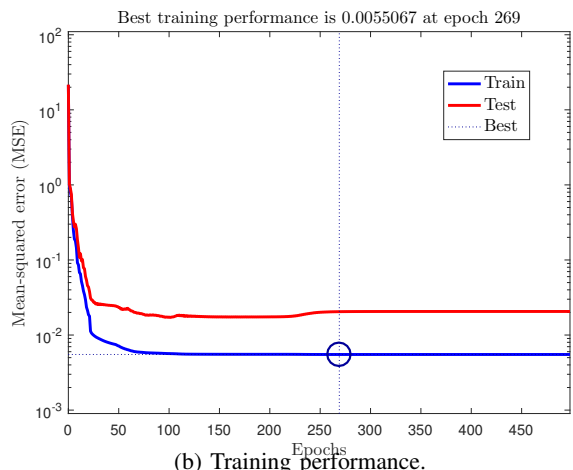}
		\vspace{-4mm}
		\caption{\small Training performance.}
		\label{fig:performance}
	\end{subfigure}
	\caption{\small Training phase of LSTM blocks.}
	\label{fig:lstmtraining}
	\vspace{-4mm}
\end{figure*}

Fig. \ref{fig:trainingout} shows the output of the LSTM trainer with $ n=10,n_s=10, $ and $ \beta=0.5 $. From Fig. \ref{fig:trainingout}, we observe that the trained output of the LSTM, which is the dynamic watermarked signal, is very close to the training target $ w $. Moreover, Fig. \ref{fig:performance} shows that the training phase converges after $ 269 $ \emph{epochs}. Here, an epoch is a measure of the number of times all the training vectors are used once to update the weights of the neural network. Since each IoTD has its own LSTM, after training the IoTDs in parallel for approximately one day, we can use and effectively authenticate the IoTDs in the system. {The training error is $  0.0055 $ which is calculated using mean squared error. Moreover, we tested the trained LSTM on another accelerometer data, and the testing error is small and close to $ 0.015 $.}

\begin{figure*}
	\begin{subfigure}{0.5\textwidth}
	\centering
	\includegraphics[width=\columnwidth]{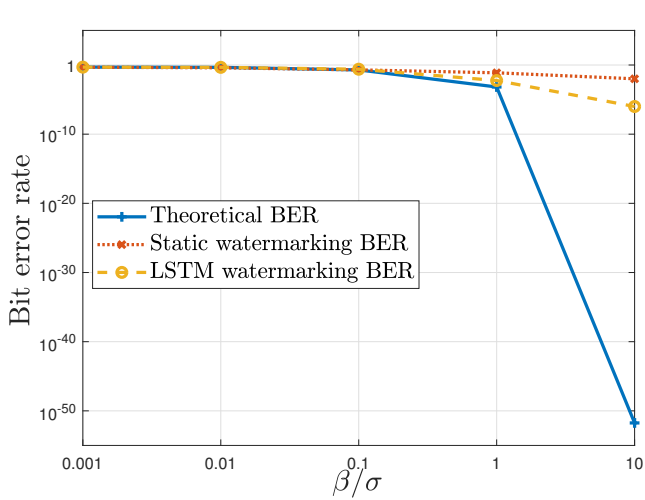}
	\caption{\small Bit error rate of proposed watermarking algorithms.}
	\label{fig:beta}
	\end{subfigure}
	\begin{subfigure}{0.5\textwidth}
	\centering
	\includegraphics[width=\columnwidth]{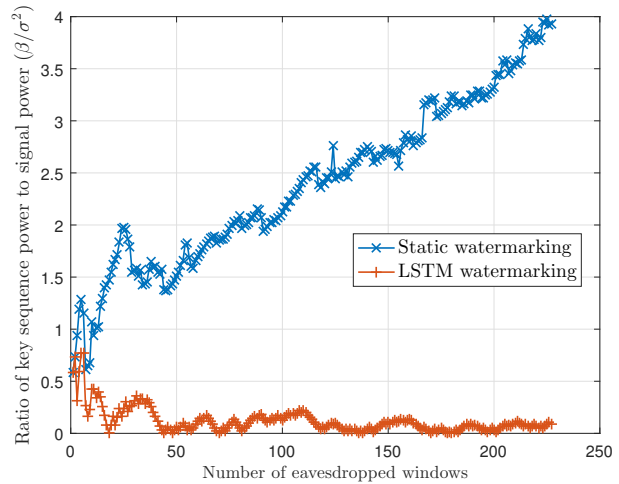}
	\vspace{-4mm}
	\caption{\small The power ratio of key to signal in presence of a {dynamic data injection} attack.}
	\label{fig:eavesdroppingpower}
	\end{subfigure}
	\caption{\small Dynamic watermarking LSTM performance.}
	\label{fig:DW-LSTM}
\end{figure*}
\begin{figure*}
	\centering
	\captionsetup{justification=justified,singlelinecheck=false}
	\includegraphics[width=\textwidth]{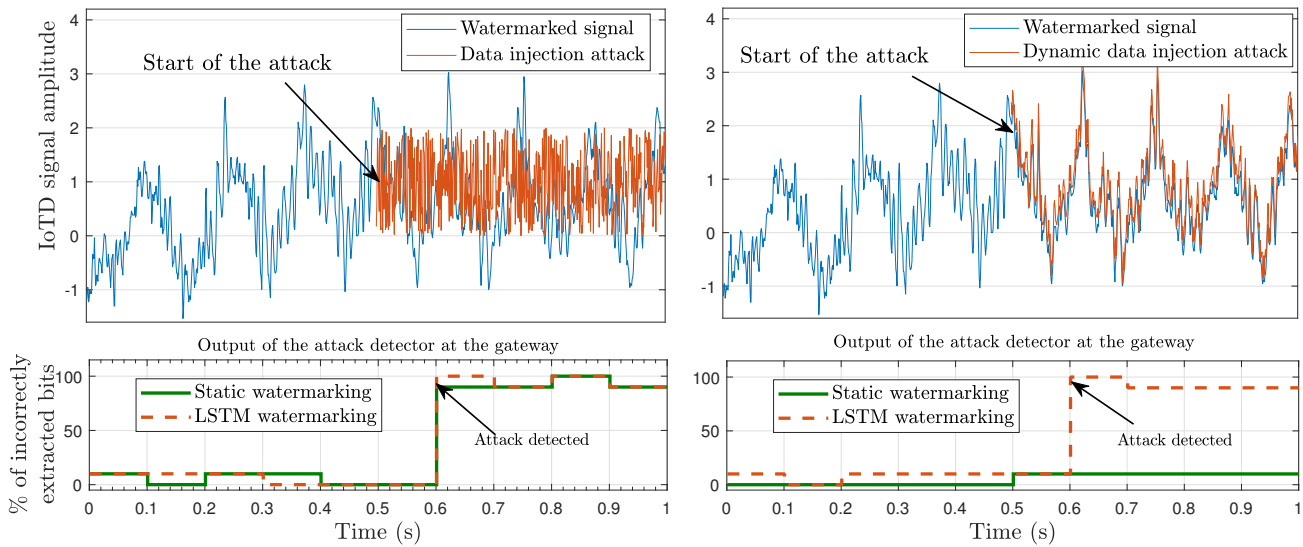}
	\caption{\small Attack detection analysis of static and LSTM watermarking schemes. The  y-axis in the lower figures shows the percentage of difference between the extracted bit stream and the hidden bit stream. This difference is used at the gateway as a metric to detect attacks on the IoTDs.}
	\label{fig:attackdetection}
\end{figure*}

Fig. \ref{fig:beta} illustrates the higher performance of LSTM compared to static watermarking in bit extraction. From \eqref{errprob}, we know that higher $ \beta/\sigma $ results in lower bit error. We can see from Fig. \ref{fig:beta} that the extraction error rate for LSTM is approximately one order of magnitude lower than static watermarking when $ \beta/\sigma=1 $. This ratio improves for higher $ \beta/\sigma $, as Fig. \ref{fig:beta} shows that the error rate of LSTM is almost two orders of magnitude lower than static watermarking when $ \beta/\sigma=10 $. This result allows designing attack detectors with lower delay, since we can choose lower $ n $ for LSTM which results in smaller window size and reduces the detection delay. In addition, Fig. \ref{fig:eavesdroppingpower} shows how a {dynamic data injection} attack operates against the two watermarking schemes. We can see that, in static watermarking, the attacker records the signal and by summing the recorded data of each window, increases the ratio of pseudo-noise key power to the signal power and extracts the bit stream. However, since in LSTM the bit stream dynamically changes in each window, the summation of the recorded data will not increase the ratio of the key power to the signal power. Therefore, the attacker will not be able to extract the bit stream and key from the recorded data. 

To analyze the effectiveness of our proposed watermarking schemes in attack detection, we choose a static watermarking block with $ n=10, n_s=10$, and $ \beta=0.5 $. {Note that, in order to choose the number of LSTM blocks, we follow a trial and error procedure by choosing different number of LSTM blocks and checking the error rate. Finally, the best performance was achieved using 100 successive LSTM blocks at the gateway and IoTDs.} We also train our LSTM blocks with these features. Then, we implement two types of attacks: a) a data injection attack in which the attacker starts to change the IoTD signal and b) a {dynamic data injection} attack, in which the attacker records the data from the IoTD, extracts the bit stream, and implements an attack with the the same watermarking bit stream. In Fig. \ref{fig:attackdetection}, the attack detectors compare the extracted bit stream with the actual hidden bit stream and the percentage of difference between these two is considered as a metric for attack detection. In other words, a high difference between the two bit streams triggers an attack detection alarm. Fig. \ref{fig:attackdetection} shows that, for the first attack, both watermarking schemes can detect the attack. However, for a {dynamic data injection attack}, static watermarking cannot detect the existence of attack while the LSTM performs well. The reason is that, in static watermarking, the bit stream is the same for all the time windows while in LSTM the watermarked bit stream dynamically changes for each time window. In addition, we can see from Fig. \ref{fig:attackdetection} that the delay of attack detection is $ 0.1 $ seconds since the attacks starts at $ 0.5 s $ and the attack detector triggers the alarm at $ 0.6 s $. The reason is that, for a window size of $ n\times n_s /f_s=0.1 $ seconds, the gateway must wait for one window to collect the IoTD data.

\begin{figure*}[!t]
	\begin{subfigure}{0.5\textwidth}
		\centering
		\includegraphics[width=\columnwidth]{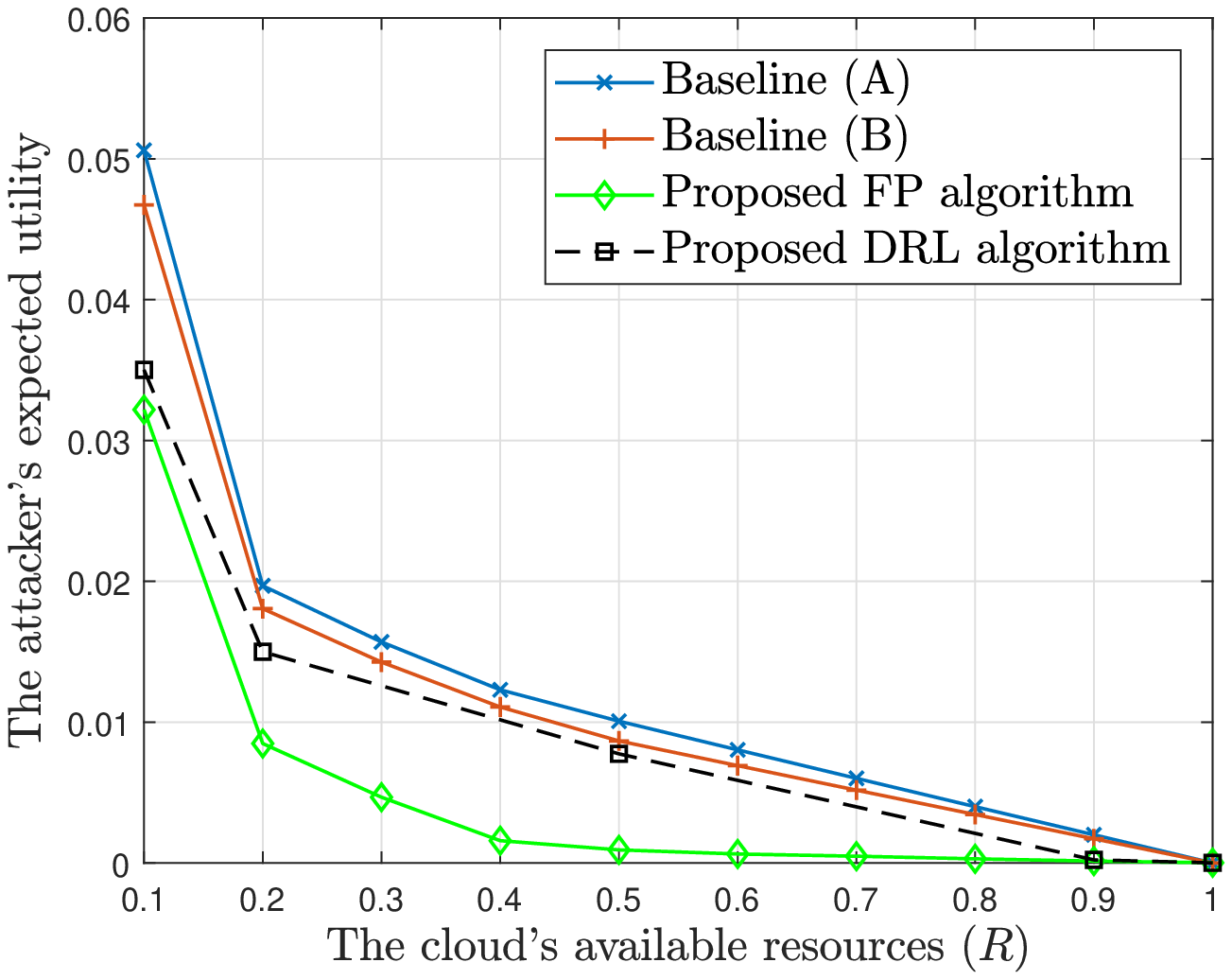}
	\end{subfigure}
	\begin{subfigure}{0.5\textwidth}
		\centering
		\includegraphics[width=\columnwidth]{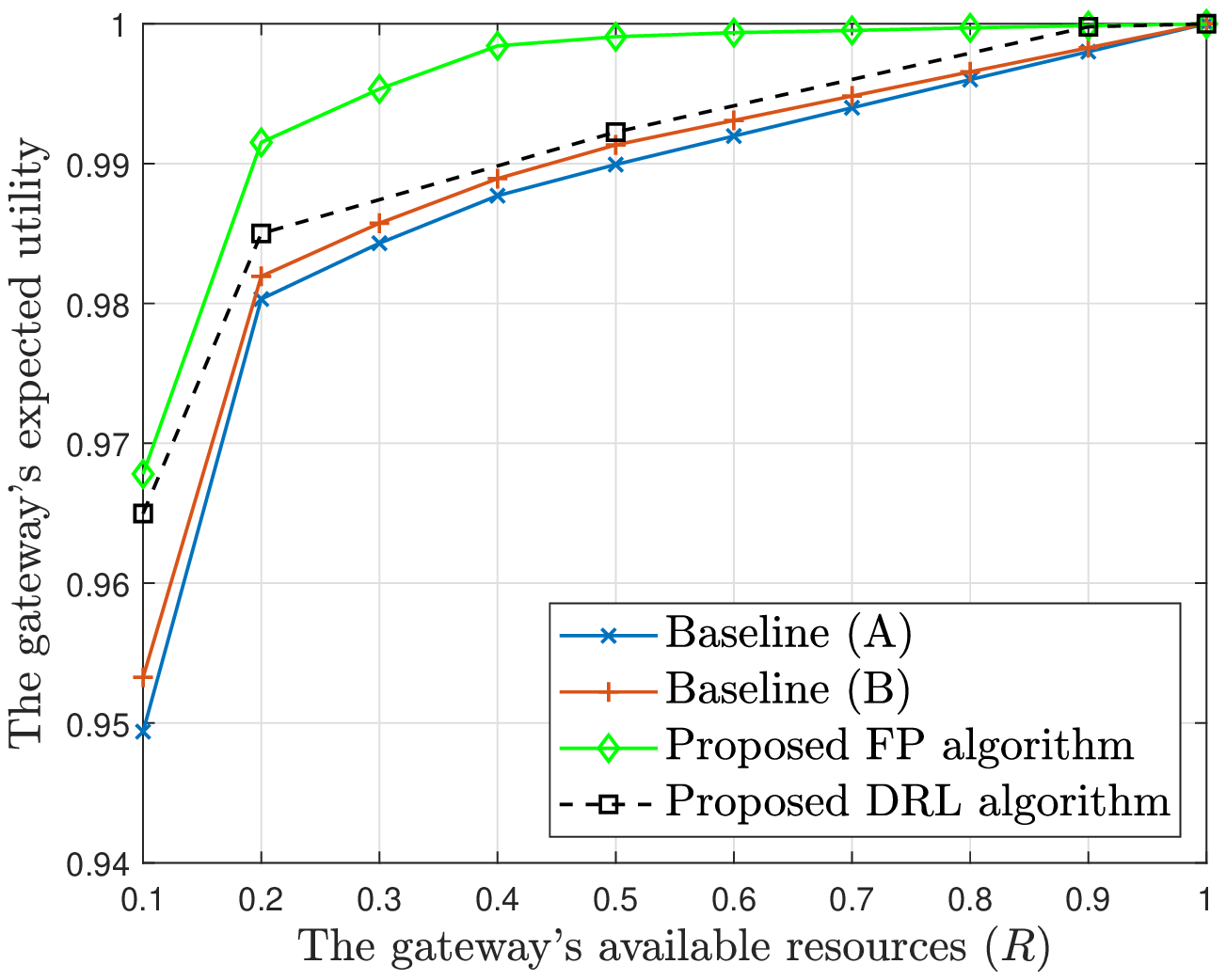}
	\end{subfigure}
	\captionsetup{justification=justified,singlelinecheck=false}
	\caption{\small The gateway's and attacker's expected utility as a function of gateway's resource, $ R $, for a scenario with $ k=100 $ and 1000 IoTDs. }
	\label{fig:limsweep}
\end{figure*}

Next, to evaluate the performance of the game-theoretic framework, we consider a system with $ 1000 $ IoTDs with sampling frequency in the range of $ [1000,15000] $ Hz. In these simulations, we compare our proposed algorithm to two baseline scenarios: baseline (A) in which the gateway chooses all the IoTDs with equal probability and baseline (B) in which the gateway chooses the IoTDs with probabilities proportional to their values. We simulate our proposed fictitious play for full information and our proposed DRL for incomplete information at the gateway.

Fig. \ref{fig:limsweep} shows a simulation in which the attacker has only $K=100$ recording devices, while for the gateway $ R $ takes values in the range $ [0.1,1] $. We can see from Fig. \ref{fig:limsweep} that, as the gateway's available resources increase from $ 0.1 $ to $ 1 $, the gateway can protect more IoTDs and thus its expected utility increases from $ 0.96 $ to $ 1 $, which indicates that when $ R=1 $, the gateway can protect all of the IoTDs. Fig. \ref{fig:limsweep} also show that, since the proposed FP algorithm converges to the MSNE, it outperforms both the baseline scenarios. For instance, for low available resources the attacker gains up to $ 40\% $ more utility which means that the attacker can compromise approximately $ 20 $ more IoTDs when the gateway chooses any one of the baseline strategies compared to the proposed fictitious play. Moreover, from Fig. \ref{fig:limsweep}, we can see that, for low available resources, the DRL algorithm yields approximately up to $ 30\% $ improvements compared to both baseline scenarios which is equivalent to $ 13 $ less compromised IoTDs.  However, it has up to $ 10\% $ less expected utility than FP, that is equivalent to $ 7 $ more compromised IoTDs, which is expected due to its operation under lack of information.

\begin{figure*}[!t]
	\begin{subfigure}{0.5\textwidth}
		\centering
		\includegraphics[width=\columnwidth]{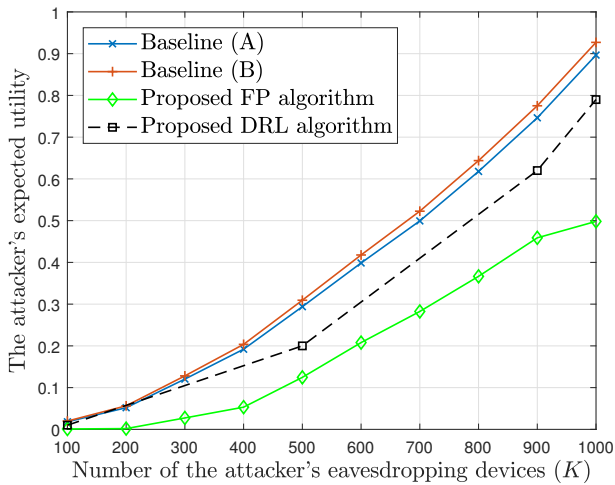}
	\end{subfigure}
	\begin{subfigure}{0.5\textwidth}
		\centering
		\includegraphics[width=\columnwidth]{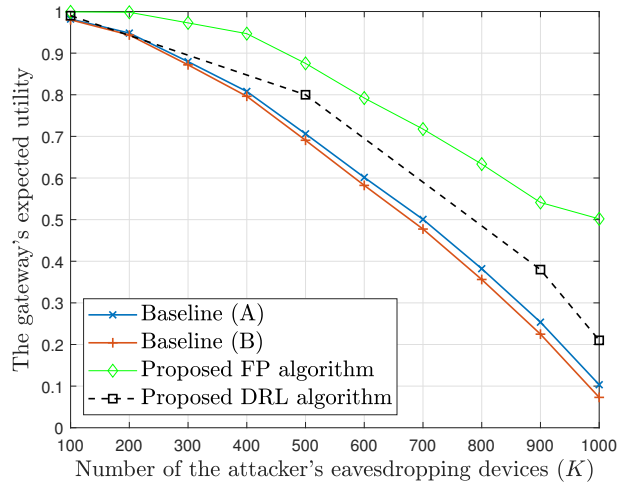}
	\end{subfigure}
	\captionsetup{justification=justified,singlelinecheck=false}
	\caption{\small The gateway's and attacker's expected utility as a function of attacker's resource, $ K $, for a scenario with $ R=0.5 $ and 1000 IoTDs. }
	\label{fig:ksweep}
\end{figure*}

Fig. \ref{fig:ksweep} shows the players' expected utility as the number of recording devices varies, for a scenario in which the available computational resources for the gateway are such that, $ R=0.5 $. Fig. \ref{fig:ksweep} shows that, as the number of attacker's recording devices increases, its expected utility increases and the attacker can disrupt more IoTDs while staying undetected. Fig. \ref{fig:ksweep} also shows that using the proposed FP, the gateway can gain a higher expected utility than the baseline scenarios. Note that, when the attacker can attack all the IoTDs, $ K=1000 $, the proposed FP algorithm has almost $ 5 $ times higher expected utility than the baseline scenarios. Furthermore, although proposed DRL has a lower expected utility than FP, it yields a $ 20\% $ improvement in the expected utility compared to the baseline scenarios for the cases where the attacker can attack all of the IoTDs.
\section{Conclusion}\label{sect:conc}
In this paper, we have proposed a novel deep learning method based on LSTM blocks for enabling detection of data injection attacks on IoTDs. We have introduced two watermarking schemes in which the IoT's gateway, which collects the data from the IoTDs, can authenticate the reliability of received signals. We have shown that our proposed LSTM method is suitable for IoTD-gateway signal authentication due to low complexity, small delay, and high accuracy in attack detection. Moreover, we have studied the massive IoT scenario in which the gateway cannot authenticate all the IoTDs simultaneously due to computational limitations. We have proposed a game-theoretic framework to address this problem, and we have derived the mixed-strategy Nash equilibrium and studied its properties. Furthermore, we have shown that analytically deriving the equilibrium is highly complicated in massive IoT scenarios, and, thus, we have proposed two learning algorithms two address this problem: a) a fictitious play algorithm that considers complete information about all the IoTDs' state and converges to the mixed-strategy Nash equilibrium and b) a deep reinforcement algorithm which predicts the set of vulnerable IoTDs for the case in which the gateway cannot have information about the unauthenticated IoTDs' state. Simulation results have shown the effectiveness of the proposed authentication schemes.
\begin{appendix}
\subsection{Proof of Theorem \ref{Theorem:hyperparameters}}\label{App:hyperparameters}
	$ \beta_i $ must be chosen such that the attacker cannot use $ w_i $ instead of $ p_i $ to extract the embedded bit. Therefore, we have to adjust $ \beta_i $ to cause a high bit error rate during the extraction of the hidden bit using the $ w $ sequence. Thus, for two different watermarked signals $ w_{i_1} $ and $ w_{i_2} $ with equal embedded bit $ b_i=1 $, we have:
	\begin{align}
	\hat{b}_i&=\frac{<w_{i_1},w_{i_2}>_{n_i}}{\beta_i^2 n_i^2}=\frac{<y_{i_1}+\beta_i b_i p_i,y_{i_2}+\beta_i b_i p_i>_{n_i}}{\beta_i^2 n_i^2}\nonumber\\
	&=\frac{<y_{i_1},y_{i_2}>_{n_i}+<y_{i_1},\beta_i b_i p_i>_{n_i}+<\beta_i b_i p_i,y_{i_2}>_{n_i}}{\beta_i^2 n_i^2}\nonumber\\&+\frac{<\beta_i b_i p_i,\beta_i b_i p_i>_{n_i}}{\beta_i^2 n_i^2}=Z_{i_1}+Z_{i_2}+Z_{i_3}+b_i^2,
	\end{align}
	where $ Z_{i_1} $, $ Z_{i_2} $, and $ Z_{i_3} $ are Gaussian random variables with distributions $ \mathcal{N}\left(\frac{\mu_{i_1}}{\beta_i^2 n_i^2},\frac{\sigma_{i_1}^2}{\beta_i^4n_i^3}\right) $, $\allowbreak \mathcal{N}\left(0,\frac{\sigma_i^2}{\beta_i^4 n_i^3}\right)$, $ \mathcal{N}\left(0,\frac{\sigma_i^2}{\beta_i^4 n_i^3}\right) $, respectively (the proof is analogous to Lemma \ref{Theorem:BER}). Then, we can write $ \hat{b}_i $ as $
	\hat{b}_i=b_i^2+Z_{i_4}=b_i+Z_{i_4},
	$ where $ Z_{i_4} \sim N\left(\frac{\mu_{i_1}}{\beta_i^2n_i^2},\frac{\sigma_{i_1}^2+2\sigma_i^2}{\beta_i^4n_i^3}\right) $. Therefore, the bit error rate incurred during the extraction of $ \hat{b}_i $ will be:
	\begin{align}
	\textrm{Pr}\left\{\hat{b}_i<0|b_i=1\right\}&=\frac{1}{2}\textrm{erfc}\left(\frac{E(\hat{b}_i)}{\sigma_{\hat{b}_i}\sqrt{2}}\right)
	\nonumber\\&=\frac{1}{2}\textrm{erfc}\left(\frac{(1+\frac{\mu_{i_1}}{\beta_i^2n_i^2})\beta_i^2n_i\sqrt{n_i}}{ \sqrt{2(\sigma_{i_1}^2+2\sigma_i^2)}}\right).\label{errprob2}
	\end{align}
Since we want to have a high bit error rate for the attacker, we can write:	$	\textrm{Pr}\left\{\hat{b}_i<0|b_i=1\right\}\geq 1-\underbar{P}.\label{errorineq}$ By using this inequelaity and \eqref{errprob2} we can find the inequality in \eqref{eq1}. Moreover we want to have a small extraction error for gateway as in \eqref{errprob}. Thus, we have: $ \textrm{Pr}\left\{\tilde{b}_i<0|b_i=1\right\}\leq \bar{\text{P}},
$ and using \eqref{errprob} we can find \eqref{eq2}. From \eqref{eq1} and \eqref{eq2}, we can derive the values for $ \beta_i $ and $ n_i $ which satisfy the security and performance requirement of the proposed watermarking scheme. Now, since we know that $ \frac{n_in_{s_i}}{{f_i^s}} $ seconds are needed to receive all the watermarked signal, then the maximum value for $ n_{s_i} $ can be found by \eqref{eq3}.
\end{appendix}

	\def\baselinestretch{1}
	\bibliographystyle{IEEEtran}
	\bibliography{references}
	\vskip -2\baselineskip plus -1fil 
\begin{IEEEbiography}[{\includegraphics[width=1in,height=1.25in,clip,keepaspectratio]{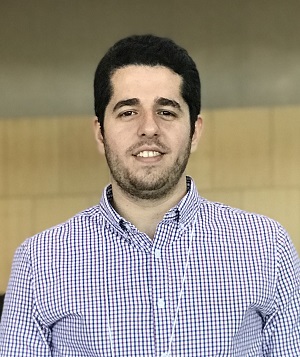}}]{Aidin Ferdowsi}
	received his BS in Electrical Engineering from the University of Tehran, Iran, in 2016. He is currently a Ph.D. student at the Bradley department of Electrical and Computer Engineering at Virginia Tech. He is also a Wireless@VT Fellow. His research interests include cyber-physical systems, machine learning, security, and game theory.
\end{IEEEbiography}
\vskip 0pt plus -1fil 
\begin{IEEEbiography}[{\includegraphics[width=1in,height=1.25in,clip,keepaspectratio]{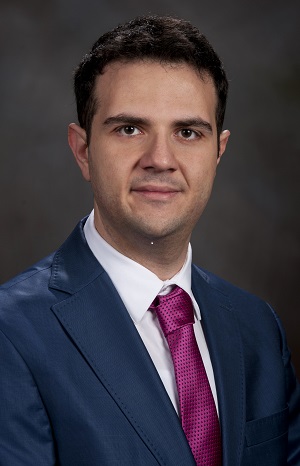}}]{Walid Saad}
	(S'07, M'10, SM’15) received his Ph.D degree from the University of Oslo in 2010. Currently,  he is an Associate Professor at the Department of Electrical and Computer Engineering at Virginia Tech, where he leads the Network Science, Wireless, and Security (NetSciWiS) laboratory, within the Wireless@VT research group. His  research interests include wireless networks, machine learning, game theory, cybersecurity, unmanned aerial vehicles, and cyber-physical systems. Dr. Saad is the recipient of the NSF CAREER award in 2013, the AFOSR summer faculty fellowship in 2014, and the Young Investigator Award from the Office of Naval Research (ONR) in 2015. He was the author/co-author of six conference best paper awards at WiOpt in 2009, ICIMP in 2010, IEEE WCNC in 2012,  IEEE PIMRC in 2015, IEEE SmartGridComm in 2015, and EuCNC in 2017. He is the recipient of the 2015 Fred W. Ellersick Prize from the IEEE Communications Society, of the 2017 IEEE ComSoc Best Young Professional in Academia award, and of the 2018 IEEE ComSoc Radio Communications Committee Early Achievement Award. From 2015-2017, Dr. Saad was named the Stephen O. Lane Junior Faculty Fellow at Virginia Tech and, in 2017, he was named College of Engineering Faculty Fellow. He currently serves as an editor for the IEEE Transactions on Wireless Communications, IEEE Transactions on Communications, IEEE Transactions on Mobile Computing, and IEEE Transactions on Information Forensics and Security.
\end{IEEEbiography}

\end{document}